%% file: Optimal_Association_in_60GHz_TON_V17.tex
\documentclass[journal, 10pt, twocolumn]{IEEEtran}
\usepackage{color}
\usepackage{xcolor}
\ifCLASSINFOpdf
  \DeclareGraphicsExtensions{.pdf,.jpeg,.png}
\else
\fi
\input{commands}
\usepackage{graphics} 
\usepackage{epsfig} 
\usepackage{times} 
\usepackage{amsmath} 
\usepackage{amssymb}  
\usepackage{mathrsfs}
\usepackage{array}
\usepackage{algorithm}
\usepackage{algorithmic}
\usepackage{changepage}

\usepackage{url}
\usepackage{array}

\usepackage[english]{babel}
\usepackage{times}
\usepackage{psfrag}
\usepackage{dsfont}
\usepackage{algorithmic}
\usepackage{algorithm}
\usepackage{setspace}
\usepackage{amssymb}
\usepackage{amsmath}
\usepackage{cite}
\usepackage{graphicx}
\usepackage{color}
\usepackage{xcolor}
\usepackage{subfigure}
\usepackage{epsfig}
\usepackage{stfloats}

\newcommand{\hh}{\hspace{-1mm}} 

\newcommand{\goodgapp}{\hspace{0.1\linewidth}} 
\newcommand{\goodgap}{\hspace{0.01\linewidth}} 

\renewcommand{\baselinestretch}{0.98}

\begin{document}
\title{Optimizing Client Association in \\60 GHz Wireless Access Networks}
\author{
        George Athanasiou,~\IEEEmembership{Member,~IEEE,}
        Pradeep Chathuranga Weeraddana,~\IEEEmembership{Member,~IEEE},
        \\Carlo Fischione,~\IEEEmembership{Member,~IEEE,}
        and Leandros Tassiulas,~\IEEEmembership{Fellow,~IEEE,}

\thanks{G. Athanasiou, P. C. Weeraddana and C. Fischione are with the Automatic Control Lab, Electrical Engineering School and Access Linnaeus Center, KTH Royal Institute of Technology, Stockholm, Sweden. E-mail: \ttfamily{\{georgioa, chatw, carlofi\}@kth.se}}
\thanks{L. Tassiulas is with Computer and Communication Engineering Department, University of Thessaly, Volos, Greece. E-mail: \ttfamily{leandros@uth.gr}}
\thanks{This work was supported by the Swedish Research Council and the EU project Hydrobionets.}}

\maketitle

\begin{abstract}
MillimeterWave communications in the 60 GHz band are considered one of the key technologies for enabling multi-gigabit wireless access. However, the high propagation loss in such a band poses major obstacles to the optimal utilization of the wireless resources, where the problem of efficient client association to access points (APs) is of vital importance. In this paper, the client association in 60 GHz wireless access networks is investigated. The AP utilization and the quality of the rapidly vanishing communication links are the control parameters. Because of the tricky non-convex and combinatorial nature of the client association optimization problem, a novel solution method is developed to guarantee balanced and fair resource allocation. A new distributed, lightweight and easy to implement association algorithm, based on Lagrangian duality theory and subgradient methods, was proposed. It is shown that the algorithm is asymptotically optimal, that is, the relative duality gap diminishes to zero as the number of clients increases. Both theoretical and numerical results evince numerous useful properties of the algorithm, such as fast convergence, scalability, time efficiency, and fair execution in comparison to existing approaches. It is concluded that the proposed solution can be applied in the forthcoming 60 GHz wireless access networks.
\end{abstract}

\begin{keywords}
60 GHz Wireless Access Networks, Resource Allocation, Association Control
\end{keywords}


\section{Introduction}\label{sec:Introduction}

MillimeterWave~(mmW) communications have recently attracted the interest of academia, industry, and standardization bodies, although the technology was invented and used many decades ago, especially in the context of military applications~\cite{Giannetti99,Smulders07}. mmW communications utilizes the part of the electromagnetic spectrum between 30 and 300~GHz, which corresponds to wavelengths from 10~mm to 1~mm~\cite{Oliver89}. The continuous development of Complementary Metal Oxide Semiconductor technology~\cite{Doan04}, with low cost and low power consumption, has enabled the use of the mmW spectrum for wireless communications, including the provisioning of quality of service (QoS) sensitive applications. 

Due to the 60~GHz great commercial potential, multiple industry-led efforts and international organizations have emerged for the standardization~\cite{IBM60}.
Examples include IEEE~802.15.3 Task Group~3c~(TG3c)~\cite{802_15_3c}, IEEE~802.11ad standardization task group~\cite{802_11ad}, the WirelessHD~\cite{WirelessHD} consortium, the Wireless Gigabit Alliance~(WiGig)~\cite{WirelessG}, and many others.
More than 5~GHz of continuous bandwidth is available in many countries worldwide, which makes 60~GHz systems particularly attractive for gigabit wireless applications such as gigabyte file transfer, wireless docking stations, wireless gigabit ethernet, wireless gaming, uncompressed high definition video transmission. Scenarios such as small-cells~\cite{Hoydis11} and mobile data offloading~\cite{Lee12}, which are nowadays strongly motivated by the increased end-user connectivity requirements and mobile traffic, can be accommodated with the use of 60~GHz radio access technology.

The 60~GHz huge bandwidth offers many benefits in terms of capacity and flexibility.
For example, even with a low spectral efficiency such as 0.4 bps/Hz, 60~GHz communication systems can provide a very high data rate of 1~Gbps, making it an ideal candidate to support simple modulation techniques~\cite{Daniels10}. In contrast, ultra-wideband (UWB) systems require roughly 2 bps/Hz to achieve 1~Gbps, but its actual deployment is limited to 400 Mbps at 1~m operating range~\cite{Geng09}. The estimated spectral efficiency for IEEE~802.11n~\cite{802_11n} communication systems is 25~bps/Hz for achieving 1~Gbps, which makes their application to high bandwidth services unacceptable in terms of cost and simple implementation~\cite{Daniels10}.
Moreover, the 60~GHz communication systems are less restricted in terms of power limits, as opposed to UWB systems~\cite{Yang04}. The regulation related to 60~GHz band allows much higher transmission power, compared to other existing wireless local area and personal systems~\cite{Doan04}.

However, there are several challenging aspects and potentialities that must still be addressed in mmW communications. On the negative side, there are severe channel attenuations. Higher power levels are necessary to overcome the high path loss. The equivalent isotropic radiated power in the 60~GHz system is approximately~10 times larger than the IEEE~802.11n~\cite{802_11n} and 30000 times larger than the UWB systems. The typical operation range of 60~GHz communications is 10 to 20~meters~\cite{802_15_3c}. For example, a path loss of 2~meters results in attenuation of approximately 74~dB. Thus directionality and blockage problems may be frequent. On the positive side, the interference levels for 60~GHz are much lower compared to the congested 2.4~GHz and 5~GHz~bands.
The compact design of 60~GHz radio allows the use of multiple-antenna solutions at the user device. Compared to 5~GHz systems, the form factor is approximately 140 times smaller and thus can be conveniently integrated into consumer electronic products. The exploitation of these unique characteristics is essential for efficient resource allocation algorithms.

In this paper we address the fundamental resource allocation problem of the client association in 60~GHz wireless access networks. Such a problem is more challenging in 60 GHz band than traditional wireless networks since the wireless channel is unstable in high frequencies and several events can violate the efficient operation of the network, such as moving obstacles that can block the communication \cite{Singh07}. Specifically, we consider the natural situation where each client has to be associated to one of the wireless APs. 
%
This gives rise to a challenging mixed integer linear optimization problem, which is combinatorial and non-convex in general and thus hard to solve efficiently.
Existing methods, such solution approaches for the \emph{generalized assignment problem} in combinatorial optimization~\cite[\S~8]{Bertsekas-98}, cannot be used because our main goal is to \emph{minimize the maximum AP utilization} in the network and ensure a \emph{fair} load distribution, which cannot be modeled by such a \emph{generalized assignment problem}. 
Nevertheless, based on Lagrangian duality theory~\cite[\S~5]{Boyd-Vandenberghe-04} and on subgradient methods, we develop a new solution approach and derive a distributed and iterative algorithm for client association (\emph{DAA}). We show the asymptotic optimality of the proposed solution method by an analytical bound on the duality gap. In addition, sufficient conditions under which strong duality holds are also characterized. The sensitivity of the convergence speed of \emph{DAA} to the variation of the numbers of clients and APs is analytically investigated. Numerical simulations illustrate and compare \emph{DAA} to benchmark algorithms, including the existing standard association approaches.

Unlike the client association approaches in more traditional access networks~\cite{Bejerano1}, we take into account the load, the channel quality and the special communication characteristics of 60~GHz radio channel and we design a dynamic association mechanism that ensures balanced and fair load distribution among the APs. Global methods~\cite{Horst-Pardalos-Toai-00} may be employed to find the solution of the combinatorial client association problem. However, global approaches have the drawbacks of 1) the prohibitive computational complexity, even in the case of problems with few variables and 2) they are inherently centralized, which makes their distributed implementation very difficult if not impossible. 
In contrast, our proposed method is \emph{fast} and can be implemented in a \emph{distributed} manner, which is crucial for many applications. Last but not least, we discuss in detail the applicability of the association algorithm on top of the current standardized medium access control (MAC) mechanisms. 
The model and the lightweight algorithm proposed in this work are general and can be applied with different existing MAC mechanisms for 60~GHz access networks, such as IEEE~802.15.3c and IEEE~802.1ad. 

The rest of the paper is organized as follows. In \S~\ref{sec:Related_work} we present related approaches in literature. A description of the system model and the problem formulation is presented in \S~\ref{sec:SysModel_mini_max_primal_problem}. In \S~\ref{sec:dual_prob}, we give the general solution approach to the client association problem by using duality theory and subgradient method. In \S~\ref{sec:algorithm_prop}, we describe the important properties of the proposed algorithm. In \S~\ref{sec:numerical_results} numerical results are presented. Lastly, \S~\ref{sec:conclusions} concludes the paper.

\section{Related work}\label{sec:Related_work}
During the last decade, resource allocation and in particular association/handoff control for wireless local area networks (WLANs) have been the focus of intense research. Several studies have analyzed the performance of the basic association policy that IEEE 802.11 standard defines based on the received signal strength indicator (RSSI). These studies have showed that this basic association policy can lead to inefficient use of the network resources \cite{Bejerano1}, \cite{Arbaugh}, \cite{Bejerano2}. Therefore, there has been increasing interest in designing better client association policies. In what follows we summarize some representative approaches in literature.

The authors in \cite{Bejerano2} study a client association policy that ensures network-wide max-min fair bandwidth allocation to the users. The work in \cite{Kauffmann} presents self-configuring algorithms that provide improved client association and fair resource sharing. In \cite{Kumar05} the problem of optimal user association to the available APs is formulated as a utility maximization problem. In \cite{Shin}, neighbor and non-overlap graphs are utilized to reduce the probing latency. Moreover, the ``multi-homing" scenario is introduced in \cite{Shakkottai}, where the traffic is split among the available APs. In this approach, the throughput is maximized by constructing a fluid model of user population that is multi-homed by the available APs in the network. The system presented in \cite{Bejerano1} ensures fairness and quality of service (QoS) provisioning in WLANs with multiple APs. The authors compare their system to 802.11 and prove that 802.11 cannot support both fairness and QoS guarantees.

Several approaches in literature jointly consider the association and handoff problems in WLANs. The authors of \cite{Mhatre} propose an association/handoff mechanism that is based on performance monitoring of the wireless links. In \cite{Arbaugh}, there is a detailed analysis of the 802.11 association/handoff process. The presented approach optimizes the probe phase of the association process to reduce the probe latency. In \cite{Amir06}, the authors propose protocols for an architecture called \emph{SMesh}, which supports fast handoff introducing collaboration between the APs in a mesh network. The work in \cite{Park07} presents a pre-authentication mechanism that speeds up the handoff process.

Another line of research considers user service requests by readjusting the load across all APs\cite{Balachandran02}. Therein, a network-directed roaming approach suggests to the clients the most appropriate network location to satisfy their requests. In \cite{Lee07} and \cite{Lee05}, a dual-association approach in wireless mesh networks is presented, where the APs for unicast traffic and the APs for broadcast traffic are independently chosen by exploiting overlapping coverage and optimizing the overall network load. Lastly, in \cite{Athanasiou07}, \cite{Athanasiou09} dynamic association and reassociation procedures are introduced with the use of the notion of the \emph{airtime cost}. This metric reflects the uplink/downlink channel conditions and the traffic load in the network. The cross-layer extension of this mechanism considers the routing-based information from the mesh backbone. Accordingly, the clients optimize their association/handoff decision.

The previous approaches are hard to apply in 60 GHz wireless access networks, due to the special characteristics of the 60 GHz channel, and the obvious differences with the rest wireless access technologies that we have previously mentioned. It follows that novel mechanisms must be designed to provide optimal resource allocation. These mechanisms must take into account the characteristics of 60 GHz wireless channel such as increased path loss, short range, fragile links, etc. Unfortunately, there is not much research in this field.

Some recent interesting approaches on 60 GHz wireless personal and local area networks have appeared in the literature. The authors in \cite{Pyo09}, study the throughput of mmWave personal networks. A combination of carrier sense multiple access (CSMA/CA) and time division multiple access (TDMA) is used. A technique for optimizing the throughput of the network is proposed, where collisions are reduced in the CSMA/CA time-slots by a private channel release time. In \cite{Singh10}, a MAC protocol that employs memory and learning to address deafness, while exploiting the reduction of interference between simultaneous transmissions, is proposed with reactive interference management. The directionality and blockage problems of mmW networks are studied in \cite{Singh07}. A cross-layer approach is presented, where a single hop transmission is preferred when line of sight (LOS) is available and a relay node is randomly selected as an alternative. In \cite{Cai10}, a resource management mechanism is proposed based on the exclusive region (ER) to exploit the spatial reuse of mmW networks. The authors in \cite{Singh11} describe an interference analysis framework that enables a quantitative evaluation of collision loss probability for a mmW mesh network with uncoordinated transmissions, as a function of the antenna patterns and spatial density of simultaneously transmitting nodes. Concurrent transmissions in 60 GHz wireless networks are studied in \cite{Qiao11} by exploiting the spatial reuse and time division multiplexing gain. It is shown that the network throughput is improved compared to single hop transmission schemes.

The current 60 GHz standardization bodies, such as IEEE 802.11ad and IEEE 802.15.3c, adopt the RSSI-based mechanism as the basic association/handoff functionality. However, RSSI is not an efficient decision metric for user association for several reasons. High RSSI values cannot univocally indicate high throughput. This is because RSSI not only depends on the distance from the APs, but also on the transmission powers of the APs. The accuracy of the RSSI-based technique is significantly affected by the high path loss, dispersion and directionality of the 60 GHz wireless channel. Moreover, since the wireless channel is a shared medium, throughput depends on the population of the cell served by the AP. An AP may become overloaded if a large number of clients are associated with it. Therefore, new metrics that better reflect the channel characteristics and the utilization of the AP are required.

In contrast to the existing work in literature, this paper considers the special characteristics of the 60 GHz channel in an optimization problem where the objective is to \emph{minimize the maximum AP utilization} in the network. Moreover, we design a lightweight distributed algorithm that balances the AP utilization by optimizing the client association process. We believe that this paper is the first to study such a fundamental resource allocation problem in 60 GHz wireless access networks, especially when complex scenarios with many users and high traffic demands are considered~\cite{Hoydis11, Lee12}. We propose a simple yet efficient solution and compare it to basic association policies, already in use in the present 60 GHz communication technologies under standardization (802.15.3c, 802.11ad). The work that is presented and evaluated in the forthcoming sections is complementary to the aforementioned resource management and scheduling approaches (the clients must first be assigned to the available APs and then the scheduling of the transmission can be handled under the established association state).

\begin{figure}[t]
\centering
\includegraphics[height=0.13\textheight]{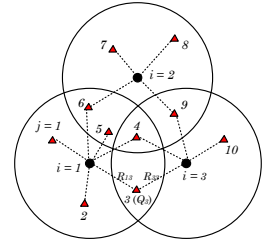}\vspace{-2mm}
\caption{Wireless access network: $\mathcal{N}=\{1,2,3\}$, $\mathcal{M}=\{1,\ldots,10\}$, $\mathcal{M}_1=\{1,\ldots,6\}$, $\mathcal{M}_2=\{4,\ldots,9\}$, $\mathcal{M}_3=\{3,4,9,10\}$, $\mathcal{N}_j=\{1\}$ for $j\in\{1,2\}$, $\mathcal{N}_j=\{2\}$ for $j\in\{7,8\}$, $\mathcal{N}_j=\{3\}$ for $j=10$, $\mathcal{N}_j=\{1,2\}$ for $j\in\{5,6\}$, $\mathcal{N}_j=\{1,3\}$ for $j=3$, $\mathcal{N}_j=\{2,3\}$ for $j=9$, $\mathcal{N}_j=\{1,2,3\}$ for $j=4$. The area inside the solid-lined circles around the APs $1,2$, and $3$ represent the transmission regions of each AP. The demanded data rate for client~$3$ is $Q_3$, the offered transmission rate from AP~$1$ to client~$3$ is $R_{13}$, and the offered transmission rate from AP~$3$ to client~$3$ is $R_{33}$.}
\label{fig:system_example}
\vspace{-4.0mm}
\end{figure}

\section{System model and problem formulation}\label{sec:SysModel_mini_max_primal_problem}

A 60 GHz wireless access network consisting of $N$ APs and $M$ clients is considered. We denote the set of APs by $\mathcal{N}=\{1,\ldots,N\}$ and the set of clients by $\mathcal{M}=\{1,\ldots,M\}$.
The set of clients that can be associated to AP~$i$ is denoted by $\mathcal{M}_i$. We assume that there are no isolated clients, i.e., $\mathcal{M}_1\cup\cdots\cup\mathcal{M}_N=\mathcal{M}$. We denote by $\mathcal{N}_j$ the set of candidate APs that client $j$ could be associated with. Figure~\ref{fig:system_example} shows an example access network, where the clients positioned inside a disc with radius $r$ (centered at the location of AP~$i$) can be associated with AP~$i$. However, the disk-shaped region is only used for illustrative purposes.

Each node (AP or client) is equipped with steerable directional antennas and it can direct its beams to transmit or to receive. We assume that AP~$i$ can support all the clients in $\mathcal{M}_i$ with a separate transmit beam. We consider the case where receivers are using single-user detection (i.e., a receiver decodes each of its intended signals by treating all other interfering signals as noise) and assume that the achievable rate from AP~$i$ to client~$j\in\mathcal{M}_i$ is
\begin{equation} \label{eq:AWGN}
R_{ij} = W\log_2\bigg (1+\frac{P_{ij}G_{ij}}{(N_0+I_{j})W}\bigg)\ ,
\end{equation}
where $W$ is the system bandwidth, $P_{ij}$ is the transmission power of AP~$i$ to client~$j$, $G_{ij}$ is the power gain from AP~$i$ to client~$j$, $N_0$ is the power spectral density of the noise at each receiver, and $I_{j}$ is the interference spectral density at client~$j$. All these assumptions are coherent with the literature and existing standards~\cite{802_11ad, 802_15_3c}. The power gain $G_{ij}$ is modeled as in~\cite{Qiao11}. In particular, we use the Friis transmission equation together with the~\emph{flat-top} transmit/recieve antenna gain model~\cite{Mudumbai09}, where a fixed gain is considered within the beamwidth and zero gain is considered outside the beamwidth of the antenna. In addition, we consider Rayleigh small-scale fading. Thus we have
\be\label{eq:gain_coefficient}
G_{ij} = \frac{G^{\mathrm{Tx}}_{ij}G^{\mathrm{Rx}}_{ij}\lambda^2\alpha_{ij}}{16\pi^2(d_{ij}/d_0)^\eta} \ , \ i\in\mathcal{N}, j\in\mathcal{M}_i \ ,
\ee
where $G^{\mathrm{Tx}}_{ij}$ is the transmit antenna gain from AP~$i$ to client~$j$, $G^{\mathrm{Rx}}_{ij}$ is the receive antenna gain from AP~$i$ to client~$j$, $\lambda$ is the wavelength, $\alpha_{ij}$ is the \emph{fading coefficient} which is an exponentially distributed random variables with unit mean to models the Rayleigh small-scale fading, $d_{ij}$ is the distance between AP $i$ and client $j$, $d_0$ is the \emph{far field reference distance} \cite{Anurag-Manjunath-Kuri-08}, $\eta$ is the path loss exponent,\footnote{$\eta\in [2,6]$ in IEEE~802.11ad networks~\cite{Geng09}.} and $I_{j}$ is the communication interference at client~$j$.
We capitalize on the well studied 60~GHz propagation characteristics~\cite{Mudumbai09}, such as highly directional transmissions with very narrow beamwidths and increased path losses due to the oxygen absorption, in order to assume that the communication interference $I_j$ is very small and does not affect significantly the achievable communication rates in the network.\footnote{In particular, a probabilistic analysis of the interference incurred due to uncoordinated transmissions is presented in~\cite{Mudumbai09}. It is shown that even uncoordinated transmission for different transmit-receive pairs leads to small collision probabilities and therefore, the links in the network can be considered as \emph{pseudo-wired}. That is, interference can essentially be ignored in MAC or higher layers design. Similar assumptions are also supported by using efficient channel allocation in the network~\cite{Athanasiou08} and efficient scheduling algorithms that supports concurrent transmissions~\cite{Qiao11}.} The achievable communication rates given in~(\ref{eq:AWGN}) are used to define the AP utilizations as described in the sequel.


We denote by $Q_j$ the \emph{demanded data rate} of client~$j$. The channel utilization between AP~$i$ and client~$j$ is denoted by $\beta_{ij}$ and is given by the ratio of $Q_{j}$ and $R_{ij}$, i.e.,
\begin{equation} \label{utilization}
\beta_{ij} = \frac{Q_j}{R_{ij}} \ .
\end{equation}
Intuitively, the channel utilization $\beta_{ij}$ gives an indication of the communications performance, in terms of the \emph{potential} loading of the communication channel between AP~$i$ and client~$j$. Thus, the sum of channel utilizations of AP~$i$ (or AP~$i$ utilization) is given by $\sum_{j \in \mathcal{M}_i} \beta_{ij}x_{ij}$, where $(x_{ij})_{j\in\mathcal{M}_i}$ are binary decision variables, which indicate the client association. In particular, for all $i\in\mathcal{N}$ and~$j\in\mathcal{M}_i$
\be\label{eq:decesion_variables}
x_{ij}= \left\{ \begin{array}{ll}
  1 & \ \ \textrm{if client $j$ is associated to AP $i$}\\
  0  & \ \ \mathrm{otherwise} \ .
   \end{array} \right.
\ee
AP~$i$ utilization is a metric reflecting the load. Our goal is to \emph{minimize the maximum AP utilization}. Specifically, the problem can be formally express~as
\begin{IEEEeqnarray}{lcl}\label{eq:mini_max_primal_problem}
\mbox{minimize} & \ \ & \displaystyle\max_{i\in\mathcal{N}}\ \textstyle\sum_{j \in \mathcal{M}_i} \beta_{ij}\ x_{ij}\IEEEyessubnumber\label{eq:mini_max_primal_problem1}\\
\mbox{subject to} & \ \  & Q_jx_{ij} \leq R_{ij}, \  i\in\mathcal{N}, j\in\mathcal{M}_i \IEEEyessubnumber\label{eq:mini_max_primal_problem2}\\
& \ \ & \textstyle\sum_{i \in \mathcal{N}_j} x_{ij}=1, \   j\in\mathcal{M} \IEEEyessubnumber\label{eq:mini_max_primal_problem3} \\
& \ \ & x_{ij}\in\{0,1\}, \ j\in\mathcal{M}, i\in\mathcal{N}_j \IEEEyessubnumber\label{eq:mini_max_primal_problem4} \ ,
\end{IEEEeqnarray}
where the variable is $(x_{ij})_{i\in\mathcal{N}, \ j\in\mathcal{M}_i}$. The main problem parameters are $(\beta_{ij})_{i\in\mathcal{N}, j\in\mathcal{M}_i}$, $(Q_j)_{j\in\mathcal{M}}$, and $(R_{ij})_{i\in\mathcal{N}, j\in\mathcal{M}_i}$\footnote{In general, client association affects the interference levels in wireless networks. However, when the characteristics of the 60 GHz wireless channel are considered (high oxygen absorption, etc) we can argue that the interference can be in the order of noise~\cite{Mudumbai09}, which in turn allows us to suppress the dependence of interference on the client association and consider it fixed~\cite{Qiao11}.}. The constraint (\ref{eq:mini_max_primal_problem2}) assure that the demand of client $j$ is less or equal to the achievable rate from AP $i$ to client~$j$. This constraint can usually be satisfied due to the huge available bandwidth of mmW channel. The constraint~(\ref{eq:mini_max_primal_problem3}) ensures that client $j$ is assigned to only one AP. The constraint~(\ref{eq:mini_max_primal_problem4}) indicates that the decision variables are binary. 

Note that problem~(\ref{eq:mini_max_primal_problem}) differentiates from the \emph{generalized assignment problem}~\cite[\S~8]{Bertsekas-98}, since the objective is to ensure \emph{fairness} in the APs load distribution. Therefore, the existing solution approaches~\cite[\S~8, \S~10]{Bertsekas-98} for the generalized assignment problem, do not apply here. Moreover, the client association problem is combinatorial, and we have to rely on exponentially complex global methods~\cite{Horst-Pardalos-Toai-00} to \emph{solve} it, unless new methods are developed. In the sequel, we present one such efficient solution approach, which, although strictly non optimal, is asymptotically optimal when $M$ grows. 

To conclude the problem formulation and motivate our solution methodology, we emphasize that our problem formulation captures the general client association problem in wireless access networks. However, our model and the solution approach exploit the unique characteristics of 60 GHz wireless access networks, including the \emph{instability of the wireless channel} in this frequency and the \emph{abrupt performance degradations} (due to blockage, etc) that would require immediate resource re-allocation actions (such as dynamic association control).


\section{Solution via dual problem}\label{sec:dual_prob}

Without loss of generality, we assume that $\beta_{ij}\leq1$ for all $i\in\mathcal{N}$ and $j\in\mathcal{M}_i$.\footnote{$\beta_{ij}>1$ means that a client demands a data rate more than what the wireless channel can provide. In this case, we can collect all the $i,j$ pairs for which $\beta_{ij}>1$ and modify the corresponding sets $\mathcal{N}_j$ and $\mathcal{M}_i$, without affecting the optimal value of the original problem~(\ref{eq:mini_max_primal_problem}). In particular, if $\beta_{ij}>1$, then we set $\mathcal{N}_j\asign\mathcal{N}_j\setminus\{i\}$ and $\mathcal{M}_i\asign\mathcal{M}_i\setminus\{j\}$. For example, suppose that $\beta_{12}>1$, and the corresponding sets $\mathcal{N}_2=\{1, 2\}$ and $\mathcal{M}_1=\{1,2,3\}$. Then, we simply modify $\mathcal{N}_j$ and $\mathcal{M}_i$ as follows: $\mathcal{N}_2=\{2\}$, $\mathcal{M}_1=\{1,3\}$.} As a result, the constrain (\ref{eq:mini_max_primal_problem2}) is redundant and can be dropped.
Thus, the standard equivalent epigraph form~\cite[\S~4.1.3]{Boyd-Vandenberghe-04} of problem~(\ref{eq:mini_max_primal_problem}) is 
\begin{IEEEeqnarray}{lcl}\label{eq:mini_max_primal_problem_epi}
\mbox{minimize} & \ \ & t\IEEEyessubnumber\label{eq:mini_max_primal_problem_epi1}\\
\mbox{subject to} & \ \  & \textstyle\sum_{j \in \mathcal{M}_i} \beta_{ij}\ x_{ij}\leq t, \ i\in\mathcal{N} \IEEEyessubnumber\label{eq:mini_max_primal_problem_epi2}\\
& \ \ & \textstyle\sum_{i \in \mathcal{N}_j} x_{ij}=1, \   j\in\mathcal{M} \IEEEyessubnumber\label{eq:mini_max_primal_problem_epi3} \\
& \ \ & x_{ij}\in\{0,1\}, \ j\in\mathcal{M}, i\in\mathcal{N}_j \IEEEyessubnumber\label{eq:mini_max_primal_problem_epi4} \ ,
\end{IEEEeqnarray}
where the variables are $t$ and ${\vec x}=(x_{ij})_{i\in\mathcal{N}, \ j\in\mathcal{M}_i}$. We denote by $p^\star$ the optimal value of the problem~(\ref{eq:mini_max_primal_problem_epi}). Note that problem~(\ref{eq:mini_max_primal_problem_epi}) is a mixed integer linear program (MILP) and is feasible. The existing MILP solvers are, of course, centralized and are typically based on global branch and bound algorithms, where the worst-case complexity grows exponentially with the problem sizes~\cite[\S~1.4.2]{Boyd-Vandenberghe-04}. Even small problems, with a few tens of variables, can take a very long time to be solved. In the sequel, we apply Lagrangian duality to problem~(\ref{eq:mini_max_primal_problem_epi}) to develop a novel solution method that is distributed and fast.

\subsection{Dual problem}
Let us first form the partial Lagrangian by \emph{dualizing} the first constraints of problem~(\ref{eq:mini_max_primal_problem_epi}). To do this we introduce multipliers $\boldsymbol{\lambda}=(\lambda_{i})_{i\in\mathcal{N}}$ for the first set of inequality constraints. Thus, the partial Lagrangian is given by
\begin{equation}\label{eq:partial_Lagrangian}
\begin{split}
L\big(t,{\vec x},\boldsymbol{\lambda}\big)& =\displaystyle t + \mathop{\textstyle{\sum}}_{i \in \mathcal{N}}\lambda_i\bigg(\mathop{\textstyle{\sum}}_{j \in \mathcal{M}_i} \beta_{ij}x_{ij}-t\bigg)\\
& =\displaystyle t\bigg(1-\mathop{\textstyle{\sum}}_{i \in \mathcal{N}} \lambda_{i}\bigg) + \mathop{\textstyle{\sum}}_{j \in \mathcal{M}}\mathop{\textstyle{\sum}}_{i \in \mathcal{N}_j} \beta_{ij}\lambda_i x_{ij}  \ ,
\end{split}
\end{equation}
where the second equality follows from rearranging the terms and from the equivalence of the following two sets:\footnote{This equivalence can be visualized by using a bipartite graph, where the nodes are the elements of two disjoint sets, the set of APs (i.e., $\mathcal{N}$) and the set of clients (i.e., $\mathcal{M}$), and the edges are the potential AP-client associations.}
\be\label{eq:set_equivalence}
\hspace{-0mm}\left\{(i,j)  \left|  i\in\mathcal{N}, j\in\mathcal{M}_i\right. \right\}{\equiv}\left\{(n,m)  \left|  m\in\mathcal{M}, n\in\mathcal{N}_m\right. \right\}\ .
\ee
Let $g(\boldsymbol{\lambda})$ denote the dual function obtained by minimizing the partial Lagrangian~(\ref{eq:partial_Lagrangian}) with respect to $t$ and $\vec x$. For notational simplicity, let us further denote by $\mathcal{X}$ the set of vectors ${\vec x}$ that satisfy the constraints~(\ref{eq:mini_max_primal_problem_epi3})-(\ref{eq:mini_max_primal_problem_epi4}) of problem~(\ref{eq:mini_max_primal_problem_epi}). In particular, $\mathcal{X}$ can be expressed as a Cartesian product of some sets $\mathcal{X}_j\subset\R^{n_j} \ , j\in\mathcal{M},$ i.e.,
\be \label{eq:CartisianProd}
\mathcal{X}= \mathcal{X}_1\times\mathcal{X}_2\times\cdots\times\mathcal{X}_M,
\ee
where $\mathcal{X}_j$ is given by
\be \label{eq:unitSimplex}
\hspace{-0mm}\mathcal{X}_j {=} \hspace{-1mm}\left\{ \hspace{-1mm}\textstyle {\vec x}_j{=}(x_{ij})_{i\in\mathcal{N}_j}\hspace{-1mm}  \left| \displaystyle\mathop{\textstyle{\sum}}_{i \in \mathcal{N}_j} x_{ij}{=}1, \  x_{ij}{\in}\{0,1\}, \ i{\in}\mathcal{N}_j \right. \right\}.
\ee
Thus, the dual function is
\begin{subequations}\label{eq:dual_function}
\begin{align}\label{eq:dual_function_}
   \hspace{-2mm}g\big(\boldsymbol{\lambda}\big) &  = \mathop{\mathop{\textstyle{ \inf}}_{t \in \R}}_{{\vec x}\in\mathcal{X}}\hspace{-0mm}L\big(t,{\vec x},\boldsymbol{\lambda}\big)\\ \label{eq:dual_function_2}
  & = \left\{\hspace{-0mm} \begin{array}{ll}
  \displaystyle \mathop{\textstyle{ \inf}}_{{\vec x}\in\mathcal{X}}\hspace{0mm} \mathop{\textstyle{\sum}}_{j \in \mathcal{M}}\mathop{\textstyle{\sum}}_{i \in \mathcal{N}_j} \beta_{ij}\lambda_i x_{ij} &  \displaystyle\textrm{$\mathop{\textstyle{\sum}}_{i \in \mathcal{N}} \lambda_{i}=1$}\\
  -\infty  &  \mathrm{otherwise} 
   \end{array} \right.\\ \label{eq:dual_function_3}
  & = \left\{\hspace{-0mm}  \begin{array}{ll}
  \displaystyle \mathop{\textstyle{\sum}}_{j \in \mathcal{M}}  \hspace{-0mm}\mathop{\textstyle{ \inf}}_{{\vec x}_j\in\mathcal{X}_j} \bigg(\mathop{ \textstyle{\sum}}_{i \in \mathcal{N}_j} \beta_{ij}\lambda_i x_{ij}\bigg) &  \displaystyle\textrm{$\mathop{\textstyle{\sum}}_{i \in \mathcal{N}} \lambda_{i}=1$}\\
  -\infty  &  \mathrm{otherwise} 
   \end{array} \right.\\ \label{eq:dual_function_4}
  &= \left\{ \begin{array}{ll}
  \displaystyle \mathop{\textstyle{\sum}}_{j \in \mathcal{M}} g_j(\boldsymbol{\lambda})& \ \ \displaystyle\textrm{$\mathop{\textstyle{\sum}}_{i \in \mathcal{N}} \lambda_{i}=1$}\\
  -\infty  & \ \ \mathrm{otherwise} \ ,
   \end{array} \right.
\end{align}
\end{subequations}
where the equality~(\ref{eq:dual_function_2}) follows from that the linear function $t(1-\sum_{i\in\mathcal{N}}\lambda_i)$ is bounded below only when it is identically zero, the equality~(\ref{eq:dual_function_3}) follows from (\ref{eq:CartisianProd})-(\ref{eq:unitSimplex}), and $g_j(\boldsymbol{\lambda})$ in~(\ref{eq:dual_function_4}) is the optimal value of the problem
\begin{equation} \label{eq:sub_problem}
\begin{array}{ll}
\mbox{minimize} & \mathop{\textstyle{\sum}}_{i \in \mathcal{N}_j} \beta_{ij}\lambda_i x_{ij}\\
\mbox{subject to} & {\vec x}_j\in\mathcal{X}_j \ ,
\end{array}
\end{equation}
with the variable ${\vec x}_j$. Even though problem~(\ref{eq:sub_problem}) is combinatorial, it has a closed-form solution given by
\be\label{eq:sub_problem_soln}
x^\star_{ij}= \left\{ \begin{array}{ll}
  1 & \ \ \textrm{$i=\displaystyle\mathop{\arg\min}_{n\in\mathcal{N}_j}\beta_{nj}\lambda_n$}\\
  0  & \ \ \mathrm{otherwise} \ .
   \end{array} \right.
\ee
and is computable very fast\footnote{If $\mathcal{I}=\mathop{\arg\min}_{n\in\mathcal{N}_j}\beta_{nj}(\lambda_n+\mu_{nj})$ is not a singleton, then an arbitrary $i\in\mathcal{I}$ is chosen.}.
Thus, the Lagrange dual problem is given by
\begin{IEEEeqnarray}{lcl}\label{eq:dual_problem}
\mbox{maximize} & \ \ & g(\boldsymbol{\lambda})=\textstyle\sum_{j \in \mathcal{M}} g_j(\boldsymbol{\lambda})\IEEEyessubnumber\label{eq:dual_problem1}\\
\mbox{subject to} & \ \  & \textstyle\sum_{i \in \mathcal{N}} \lambda_{i}=1 \IEEEyessubnumber\label{eq:dual_problem2}\\
& \ \ & \textstyle\lambda_i\geq 0, \ i\in\mathcal{N} \IEEEyessubnumber\label{eq:dual_problem3}  \ ,
\end{IEEEeqnarray}
where the variables is $\boldsymbol{\lambda}$. We denote by $d^\star$ the optimal value of the problem~(\ref{eq:dual_problem}) that will be useful later. Note that the Lagrange dual problem~(\ref{eq:dual_problem}) is a convex optimization problem, even though the primal problem~(\ref{eq:mini_max_primal_problem_epi}) is \emph{not} convex~(see \cite[\S~5.2]{Boyd-Vandenberghe-04}). Let us next focus on the dual problem~(\ref{eq:dual_problem}) and its solution method, which allows us to find a good feasible solution to the original problem~(\ref{eq:mini_max_primal_problem_epi}).

\subsection{Solving the dual problem via projected subgradient method}\label{subsec:soln}
The objective function $g(\boldsymbol{\lambda})$ is, in general, a non-smooth (therefore non-differentiable) concave function. A common approach to handle such non-differentiable functions is the subgradient method~\cite{Boyd-EE364b-SubGradMethods-07}, because gradient based algorithms cannot be applied. Therefore, the projected subgradient method~\cite{Bertsekas-99,Boyd-EE364b-SubGradMethods-07} is used to solve the dual problem~(\ref{eq:dual_problem}).

\begin{figure}[t]
\noindent \hrulefill
{\label{Association}

\vspace{-1mm}{ \scriptsize \textbf{\emph{DAA}: Distributed algorithm for client association }}}
\begin{enumerate}
\scriptsize
\item[1] Initialization: The local channel utilizations, i.e., $(\beta_{ij})_{j\in\mathcal{M}_i}$ are given, at every AP $i$. Set subgradient iteration index $k=1$. Each AP~$i$ broadcasts the initial \emph{feasible} prices $\lambda^{(k)}_i$ to its local clients $j\in\mathcal{M}_i$.

\item[2] Every client $j$ sets $\boldsymbol\lambda=\boldsymbol\lambda^{(k)}$ and locally determines its association by solving problem~(\ref{eq:sub_problem}). Denote by AP~$i_j$ the AP for which $x_{ij}=1$. 

\item[3] Client $j~(\in\mathcal{M})$ signals \emph{only} to AP~$i_j$ an do not send any signalling to other AP~$i$~s, where $i\in\mathcal{N}_j \setminus \{i_j\}$.

\item[4] Every AP~$i$ computes $u_i$ by summing $\beta_{ij}$ over the clients~$j\in\mathcal{M}_i$, who had signalled in step~3, see~(\ref{eq:subgrad_elements}).

\item[5] Subgradient iteration: APs communicate and form ${\vec u}^{(k)}$ by combining each $u_i$ and by performing~(\ref{eq:prjected_subgrad_method}) to compute $\boldsymbol\lambda^{(k+1)}$.

\item[6] Stopping criterion: if the stopping criterion is satisfied, STOP. Otherwise, set $k=k+1$, each AP~$i$ broadcasts the \emph{feasible} prices $\lambda^{(k)}_i$ to its local clients $j\in\mathcal{M}_i$ and go to step~2.
\end{enumerate}
\vspace{-3mm}
\noindent\hrulefill\vspace{-6mm}
\end{figure}

First, we denote by $\vec u$ a subgradient of $-g$ at a feasible $\boldsymbol{\lambda}$, where ${\vec u}=(u_i)_{i\in\mathcal{N}}$.
Specifically,
\be\label{eq:subgrad_elements}
\textstyle u_i = -\sum_{j\in\mathcal{M}_i}\beta_{ij}x^\star_{ij} \ ,
\ee
where $x^\star_{ij}$ for all $j\in\mathcal{M}$ and $i\in\mathcal{N}_j$ is obtained as the solution of problem~(\ref{eq:sub_problem}) for all $j\in\mathcal{M}$. Thus the projected subgradient method is given by
\be\label{eq:prjected_subgrad_method}
\boldsymbol{\lambda}^{(k+1)} =  P\big(\boldsymbol{\lambda}^{(k)} - \alpha_k {\vec u}^{(k)}\big) \ ,
\ee
where $k$ is the current iteration index of the subgradient method, $\alpha_k>0$ is the $k$th step size\footnote{We chose \emph{square summable but not summable} step size (e.g., $\alpha_k=a/k$, where $0< a<\infty$), that guarantees the asymptotic convergence of the subgradient method~\cite{Boyd-EE364b-SubGradMethods-07}.}, and $P$ is Euclidean projection onto the unit simplex $\Pi=\{\boldsymbol{\lambda} \left| \ \sum_{i \in \mathcal{N}} \lambda_{i}=1, \lambda_i\geq0 \right.\}$ (see~\cite[Excercise~2.1.12]{Bertsekas-99}). By employing (\ref{eq:prjected_subgrad_method}) in an iterative manner, we can \emph{solve} the dual problem~(\ref{eq:dual_problem}). However, recovering a primal feasible solution is nontrivial because the original problem~(\ref{eq:mini_max_primal_problem_epi}) is noncovex. A discussion of these nontrivial issues and how to find a good feasible solution is deferred to~\S~\ref{subsec:RecoverPrimal}, to maintain a cohesive presentation. Let us next describe how the computation of the solution of problem~(\ref{eq:dual_problem}) is performed in a distributed manner.


\subsection{Distributed algorithm for client association (\emph{DAA})}\label{subsec:distributed_alg}
Recall that the dual function $g(\boldsymbol\lambda)$ is separable among the clients $j\in\mathcal{M}$, see the objective function~(\ref{eq:dual_problem1}) of problem~(\ref{eq:dual_problem}). Therefore, the subgradient components~(\ref{eq:subgrad_elements}) for the subgradient method (\ref{eq:prjected_subgrad_method}) can be computed by coordinating the problem~(\ref{eq:sub_problem}) for all $j\in\mathcal{M}$. This suggests \emph{DAA} presented at the top of this~page.

The first step initializes \emph{DAA}. Step~2 represents the optimization performed in a decentralized fashion by each client for fixed $\boldsymbol \lambda$. The optimization at each client~$j$ is a very simple operation and is to find the AP~$i_j$, where, $i_j=\arg\min_{i\in\mathcal{N}_j}\beta_{ij}\lambda_i$~(see~(\ref{eq:sub_problem_soln})). Step~3 requires signaling between clients and APs. In particular, each client~$j$ signals only to AP~$i_j$. This signalling process can be performed very efficiently, e.g., binary signaling. As a result, we have a light protocol between clients and APs. In Step~4, each AP~$i$ locally computes $u_i$~(see~\ref{eq:subgrad_elements}), which is the summation of $\beta_{ij}$ over the client who signalled the AP. Step~5 requires AP coordination. In particular, APs coordinate to perform~(\ref{eq:prjected_subgrad_method}), which is the projection of a point onto the unit simplex. The result of this operation is given by the solution to a convex optimization problem, which can be carried out efficiently. Step~6, is the stopping criterion for the algorithm. If the stopping criterion is satisfied, \emph{DAA} terminates. Otherwise, the algorithm continues in an iterative manner. In practice, a natural stopping criterion would be running it for a fixed number of iterations.

\subsection{Distributed implementation over existing standards}\label{subsec:DAA_over_standards}

Let us discuss now how the actual implementation of the proposed \emph{DAA} algorithm could be achieved on top of the existing standards, IEEE 802.15.3c and IEEE 802.11ad. The distributed algorithm is performed periodically in the system to ensure the balanced operation of the network. The period of the execution is given by the control messages established by the medium access control protocol, as we see in detail~below.

The iterative association algorithm does not have to be executed every time a client initiates an association or a handoff process. We assume that the \emph{newcomer} client follows the association mechanism that IEEE 802.15.3c and IEEE 802.11ad define, based on the RSSI. Then, our algorithm is periodically executed to \emph{correct} possible suboptimal client associations in the network by reallocating the available resources. As mentioned before the distributed nature of the association algorithm is crucial, in the direction of offloading the APs and make good use of the small computational resources that the clients may provide. Both 802.15.3c and 802.11ad define control frames (denoted as beacon frames) that are periodically broadcasted by the APs in the network. The APs can utilize these frames to trigger the initialization of \emph{DAA} and carry the required information to the clients. The APs inform their clients about the initialization of \emph{DAA} by setting a special bit into the beacon frame. Thus, the clients are ready to cooperate towards the optimal resource allocation in the network. The information required by the algorithm can be carried in the control frames or piggy-backed to the data frames that the APs send to the clients \cite{Athanasiou08}. Moreover, the clients are piggy-backing the information in the data frames that they send to the APs. Thus, the algorithm is executed in perfect harmony with the networking protocols, without interrupting the actual network operation (data communication) and without causing extra delays. 

\section{Algorithm Properties}\label{sec:algorithm_prop}
In this section, we first show the convergence performance for the proposed algorithm. Then we show how to recover a \emph{good} primal feasible solution. Next, we highlight some sufficient conditions under which strong duality holds for the MILP~(\ref{eq:mini_max_primal_problem_epi}) followed by a couple of examples. Finally, we show analytically the \emph{asymptotic optimality} of the algorithm, where the \emph{relative duality gap} diminishes to zero as the number of clients in the system grows.

Recall that $p^\star$ is the optimal value of the original MILP~(\ref{eq:mini_max_primal_problem_epi}) and $d^\star$ is the optimal value of the associated dual problem~(\ref{eq:dual_problem}). We refer to $p^\star$ as the \emph{primal optimal value}, $d^\star$ as the \emph{dual optimal value}, $(p^\star-d^\star)$ as the \emph{optimal duality gap}, and $(p^\star-d^\star)/p^\star$  as the \emph{optimal relative duality gap}, which are useful in the rest of the paper. 

\subsection{Convergence}\label{subsec:Convergence}
\emph{DAA} essentially solves the dual problem~(\ref{eq:dual_problem}) by using the projected subgradients method and the convergence of the algorithm is established by the following proposition:
\begin{prop}\label{prop:convergence}
Let $g^{(k)}_{\mathrm{best}}$ denote the \emph{best} dual objective value found after $k$ subgradient iterations, i.e., $g^{(k)}_{\mathrm{best}}=\max\{g(\boldsymbol\lambda^{(1)}),\ldots,g(\boldsymbol\lambda^{(k)})\}$. Then, $\forall\epsilon>0$ $\exists n\geq 1$ such that $\forall k$ $k\geq n\Rightarrow \big(d^\star-g^{(k)}_{\mathrm{best}}\big)<\epsilon$.
\end{prop}
\begin{proof}
The proof is built on the material presented in \cite[\S~3.2]{Boyd-EE364b-SubGradMethods-07}\cite{Bertsekas-99}. Let us denote by ${\boldsymbol\lambda}^\star$ the optimal solution of dual problem~(\ref{eq:dual_problem}). Thus, we have
\begin{subequations}
\begin{align}\label{eq:Convergence_1}
  ||{\boldsymbol\lambda}^{(k+1)}-{\boldsymbol\lambda}^\star||^2_2 & = ||P\big(\boldsymbol{\lambda}^{(k)} - \alpha_k {\vec u}^{(k)}\big)-{\boldsymbol\lambda}^\star||^2_2\\ \label{eq:Convergence_2}
  & \leq ||\big(\boldsymbol{\lambda}^{(k)} - \alpha_k {\vec u}^{(k)}\big)-{\boldsymbol\lambda}^\star||^2_2\\ \label{eq:Convergence_3}
  & = ||\boldsymbol{\lambda}^{(k)} -{\boldsymbol\lambda}^\star||^2_2-2\alpha_k{{\vec u}^{(k)\mbox{\scriptsize T}}}(\boldsymbol{\lambda}^{(k)} -{\boldsymbol\lambda}^\star)\nonumber \\
  & \hspace{4mm} + \alpha^2_k||{\vec u}^{(k)}||^2_2\\ \label{eq:Convergence_4}
  & \leq ||\boldsymbol{\lambda}^{(k)} -{\boldsymbol\lambda}^\star||^2_2-2\alpha_k(g({\boldsymbol\lambda}^\star)-g(\boldsymbol{\lambda}^{(k)}) ) \nonumber \\
  & \hspace{4mm} + \alpha^2_k||{\vec u}^{(k)}||^2_2\\ \label{eq:Convergence_5}
  & =||\boldsymbol{\lambda}^{(k)} -{\boldsymbol\lambda}^\star||^2_2-2\alpha_k(d^\star-g(\boldsymbol{\lambda}^{(k)}) ) \nonumber \\
  & \hspace{4mm} + \alpha^2_k||{\vec u}^{(k)}||^2_2 \ ,
\end{align}
\end{subequations}
where~(\ref{eq:Convergence_1}) follows from~(\ref{eq:prjected_subgrad_method}), (\ref{eq:Convergence_2}) follows from that the projection onto unit simplex $\Pi$ always decrease the distance of a point to every point in $\Pi$ and in particular to the optimal point ${\boldsymbol\lambda}^\star$, (\ref{eq:Convergence_4}) follows from the definition of subgradient, i.e., $-g(\boldsymbol\lambda^\star)\geq -g(\boldsymbol\lambda^{(k)}) + {\vec u}^{(k)\mbox{\scriptsize T}}(\boldsymbol\lambda^\star-\boldsymbol\lambda^{(k)})$, and (\ref{eq:Convergence_5}) follows from that $d^\star = g(\boldsymbol\lambda^\star)$. Recursively applying~(\ref{eq:Convergence_5}) and rearranging the terms, we get
\begin{subequations}
\begin{align}\label{eq:Convergence2_1}
 \hspace{-2mm}2\displaystyle \mathop{\textstyle{\sum}}_{l =1}^{k} \alpha_l(d^\star-g(\boldsymbol\lambda^{(l)})) & \textstyle = -||{\boldsymbol\lambda}^{(k+1)}-{\boldsymbol\lambda}^\star||^2_2+||{\boldsymbol\lambda}^{(1)}-{\boldsymbol\lambda}^\star||^2_2\nonumber \\
 &\hspace{4mm}+\displaystyle \mathop{\textstyle{\sum}}_{l =1}^{k}\alpha^2_l||{\vec u}^{(k)}||^2_2\\ \label{eq:Convergence2_2}
  & \textstyle \leq R^2 + G^2 \sum_{l=1}^{k}\alpha^2_l\ ,
\end{align}
\end{subequations}
where (\ref{eq:Convergence2_2}) follows from that $||{\boldsymbol\lambda}^{(k+1)}-{\boldsymbol\lambda}^\star||_2\geq0$, $||\tilde{\boldsymbol\lambda}-{\boldsymbol\lambda}^\star||_2\leq R=\sqrt{2}$ for any $\tilde{\boldsymbol\lambda}\in\Pi$, and the norm of any subgradient ${\vec u}$ of $-g$  at any $\tilde{\boldsymbol\lambda}\in\Pi$ is bounded, i.e.,\footnote{Compare to~(\ref{eq:subgrad_elements}).}
\be\label{eq:normBound}
\hspace{-0mm}||{\vec u}||_2\leq G=\sqrt{\textstyle\sum_{i\in\mathcal{N}}\big(\sum_{j\in\mathcal{M}_i}\beta_{ij}\big)^2}. \quad
\ee
Moreover, clearly we have
\be\label{eq:Convergence3_1}
d^\star-g^{(k)}_{\mathrm{best}}\leq d^\star-g(\boldsymbol\lambda^{(l)}), \ l=1,\ldots,k \ .
\ee
Thus, from~(\ref{eq:Convergence2_2}), (\ref{eq:Convergence3_1}), and by noting that step size $\alpha_l=a/l, \ 0<a<\infty$ is \emph{square summable} \big(i.e., $\sum_{l=1}^\infty\alpha^2_l=a^2\pi/6$\big), we obtain an upper bound on $d^\star-g^{(k)}_{\mathrm{best}}$ as
\begin{subequations}\label{eq:Convergence4_1}
\begin{align}\label{eq:Convergence4_11}
\textstyle d^\star-g^{(k)}_{\mathrm{best}} & \leq \textstyle(R^2 + G^2 \sum_{l=1}^{k}\alpha^2_l)/( 2\sum_{l=1}^{k}\alpha_l)\\ \label{eq:Convergence4_12}
& \leq  \textstyle\big(R^2/2 + a^2G^2 \pi/12\big)/(\sum_{l=1}^{k}\alpha_l) \ .
\end{align}
\end{subequations}
Since $\sum_{l=1}^{k}\alpha_l$ is strictly monotonically increasing in $k$ (it grows without bound as $k\rightarrow\infty$), for any $\epsilon >0$ we can always find a integer $n\geq 1$ such that $\sum_{l=1}^k\alpha_l> \epsilon \ (R^2/2 +a^2G^2 \pi/12)$ if $k\geq n$, which concludes the proof.
\end{proof}

The bound derived in~(\ref{eq:Convergence4_12}) allows us to predict some key behaviors of the convergence of the proposed algorithm. To see this, we note from~(\ref{eq:normBound}) that the numerator of the bound depends on $(\beta_{ij})_{i\in\mathcal{N}, \ j\in\mathcal{M}_i}$ for fixed $a$. Now suppose that the number of clients increases. This forces $G$ to increase as well. As a result, the corresponding total iterations to reach the given accuracy $\epsilon$ will also grow. Roughly speaking, this means that, for fixed number of APs, the larger the number of clients is, the larger the total number of iterations required for the convergence of \emph{DAA}. On the other hand, suppose that the user distribution is such that $\sum_{j\in\mathcal{M}_i}\beta_{ij}$ is roughly the same for each AP~$i$. Thus, if the total number of APs is increased, then $G$ will become larger and as a result, the total number of iterations to convergence is increased as well. These algorithm behaviors are numerically illustrated in \S~\ref{sec:numerical_results}.

\vspace{-3mm}
\subsection{Recovering a feasible primal point}\label{subsec:RecoverPrimal}
As we discussed in \S~\ref{subsec:Convergence}, we can \emph{solve} the dual problem to any given accuracy to yield the dual optimal value $d^\star$ and the dual optimal solution~$\boldsymbol\lambda^\star$. If the primal problem is convex, from $d^\star$ and $\boldsymbol\lambda^\star$, we can usually obtain the primal optimal value $p^\star$ and primal optimal solution $(t^\star, {\vec x}^\star)$~\cite[\S~5.5.5]{Boyd-Vandenberghe-04}. However, recall that the original MILP~(\ref{eq:mini_max_primal_problem_epi}) is a \emph{nonconvex} problem. Therefore, unlike convex problems, there is \emph{no guarantee} that from $d^\star$ and $\boldsymbol\lambda^\star$, we obtain $p^\star$ and $(t^\star, {\vec x}^\star)$.\footnote{The first component $t^\star$ of the primal optimal solution of MILP~(\ref{eq:mini_max_primal_problem_epi}) and the primal optimal value $p^\star$ are clearly the same.} Nevertheless, in the case of MILP~(\ref{eq:mini_max_primal_problem_epi}), a \emph{primal feasible point} is obtained during each iteration~$k$ of the algorithm~(see step~2 of the algorithm). Thus, it is natural to go for the best choice, among all the primal feasible points obtained so far. For example, a good approximation for the primal optimal value would be
\be\label{eq:best_primal_objective_afrer_k_itr}
p^{(k)}_{\mathrm{best}}=\min\{t^{(1)},\ldots,t^{(k)}\} \ ,
\ee
where $(t^{(k)}, {\vec x}^{(k)})$ is the primal feasible point in iteration~$k$.\footnote{APs can compute $(t^{(k)}, {\vec x}^{(k)})$ easily by using the client signaling they received at step~2 and the AP coordination at step~5. For example, $t^{(k)}=||-{\vec u}^{(k)}||_{\infty}$.}
A good feasible point would be $(t^{(k)}_{\mathrm{best}}, {\vec x}^{(k)}_{\mathrm{best}})$, which is the primal feasible point that corresponds to $p^{(k)}_{\mathrm{best}}$.

Even though the value $p^{(k)}_{\mathrm{best}}$ is not usually as good as the primal optimal value $p^\star$, monte Carlo simulations show that it is a good approximate value with $k$ on the order of hundreds or more, e.g., $k\geq 100$ (see \S~\ref{sec:numerical_results}). There is no clear analytical explanations of these fortuitous encounters, especially because the original problem~(\ref{eq:mini_max_primal_problem_epi}) is nonconvex~\cite[\S~6.3]{Bertsekas-99}.

\vspace{-3mm}
\subsection{Duality gap}\label{subsec:duality_gap}
The \emph{duality gap} $(p^\star-d^\star)$ is one of the important metric that can be used to quantify the performance of the proposed \emph{DAA}. Note that, in general, the duality gap for MILP~(\ref{eq:mini_max_primal_problem_epi}) is not zero, because the problem is nonconvex. Therefore, it is not surprising that deriving general conditions under which the strong duality for MILP~(\ref{eq:mini_max_primal_problem_epi}) is very difficult. Nevertheless, we first provide some examples to highlight sufficient conditions for strong duality for MILP~(\ref{eq:mini_max_primal_problem_epi}). The latter part of this section derives an analytical bound on the duality gap. Moreover, we show the \emph{asymptotic optimality} of the algorithm, where the \emph{relative duality gap} $(p^\star-d^\star)/p^\star$ diminishes to zero as the number of clients in the system grows. Such asymptotic results are indeed important from a theoretical, as well as from a practical perspective, see for example the duality results associated with the well known Knapsack problem~\cite{Bertsekas-99}.

The following proposition establishes a simple result, which is instrumental to study zero duality.

\begin{prop}\label{prop:dualoptimal_Vs_relaxed_Primal}
Let $p^\star_{\mathrm{relax}}$ denote the optimal value of the linear programming~(LP) relaxation of
problem~(\ref{eq:mini_max_primal_problem_epi}), i.e.,
\begin{IEEEeqnarray}{lcl}\label{eq:mini_max_primal_problem_epi_relax}
\mbox{minimize} & \ \ & t\IEEEyessubnumber\label{eq:mini_max_primal_problem_epi_relax1}\\
\mbox{subject to} & \ \  & \textstyle\sum_{j \in \mathcal{M}_i} \beta_{ij}\ x_{ij}\leq t, \ i\in\mathcal{N} \IEEEyessubnumber\label{eq:mini_max_primal_problem_epi_relax2}\\
& \ \ & \textstyle\sum_{i \in \mathcal{N}_j} x_{ij}=1, \   j\in\mathcal{M} \IEEEyessubnumber\label{eq:mini_max_primal_problem_epi_relax3} \\
& \ \ &  0\leq x_{ij}\leq 1, \ j\in\mathcal{M}, i\in\mathcal{N}_j \IEEEyessubnumber\label{eq:mini_max_primal_problem_epi_relax4} \ ,
\end{IEEEeqnarray}
with variables $t$ and ${\vec x}=(x_{ij})_{i\in\mathcal{N}, \ j\in\mathcal{M}_i}$. Then $d^\star = p^\star_{\mathrm{relax}}$.
\end{prop}
\begin{proof}
Note that problem~(\ref{eq:mini_max_primal_problem_epi}) is always feasible. The proof is based on two key results:
\begin{enumerate}
\item[(a)] for problem~(\ref{eq:mini_max_primal_problem_epi_relax}), we always have strong duality, i.e., $p^\star_{\mathrm{relax}}=d^\star_{\mathrm{relax}}$, and
\item[(b)] the dual of problem~(\ref{eq:mini_max_primal_problem_epi_relax}) is identical to problem~(\ref{eq:dual_problem}), and therefore $d^\star_{\mathrm{relax}}=d^\star$, where $d^\star_{\mathrm{relax}}$ denotes the dual optimal value of problem~(\ref{eq:mini_max_primal_problem_epi_relax}).
\end{enumerate}
In particular, (a) is guaranteed from strong duality results for linear programs, see~\cite[\S~5.2.3]{Boyd-Vandenberghe-04}. To prove (b), we first note the following: the partial Lagrangian obtained by dualizing first constraint of problem~(\ref{eq:mini_max_primal_problem_epi_relax}) is identical to~(\ref{eq:partial_Lagrangian}), where we consider same notations for the dual variables. This is a slight abuse of notation, but it helps in the clarity of the exposition. Let $h({\boldsymbol \lambda})$ denote the dual function obtained by minimizing the partial Lagrangian as given by~[compare to~(\ref{eq:dual_function})]
\begin{equation}\label{eq:dual_function_relaxed}
h({\boldsymbol \lambda}) =  \left\{ \begin{array}{ll}
  \displaystyle \mathop{\textstyle{\sum}}_{j \in \mathcal{M}} h_j(\boldsymbol{\lambda})& \ \ \textrm{$\mathop{\textstyle{\sum}}_{i \in \mathcal{N}} \lambda_{i}=1$}\\
  -\infty  & \ \ \mathrm{otherwise} \ ,
   \end{array} \right.
\end{equation}
where $h_j(\boldsymbol{\lambda})$ is the optimal solution of a problem very similar to (\ref{eq:sub_problem}), except that the constraint ${\vec x}_j\in\mathcal{X}_j$ is replaced by ${\vec x}_j\in\conv(\mathcal{X}_j)$. Note that $\conv(\mathcal{X}_j)\in\R^{n_j}$ is a \emph{unit simplex}, and therefore the optimal value $h_j(\boldsymbol{\lambda})$ is attained at one of the vertexes ${\vec x}_j=(x_{ij})_{i\in\mathcal{N}_j}$ of the unit simplex~\cite[Corollary~32.3.4]{Rockafellar-70}. Specifically, the components of the vertex ${\vec x}_j$ is identically given by~(\ref{eq:sub_problem_soln}). As a result, $h_j(\boldsymbol{\lambda})=g_j(\boldsymbol{\lambda})$ for all $j\in\mathcal{M}$ and $h({\boldsymbol \lambda})=g({\boldsymbol \lambda})$.
\end{proof}

Note that if problem~(\ref{eq:mini_max_primal_problem_epi_relax}), the LP relaxation of problem~(\ref{eq:mini_max_primal_problem_epi}), has integer solutions, then we can easily show that the optimal value $p^\star$ of the original MILP~(\ref{eq:mini_max_primal_problem_epi}) is equal to the optimal value $p^\star_{\mathrm{relax}}$, i.e., $p^\star =p^\star_{\mathrm{relax}}$. Therefore, from \emph{Proposition~\ref{prop:dualoptimal_Vs_relaxed_Primal}}, we have $p^\star =d^\star$. In other words, if problem~(\ref{eq:mini_max_primal_problem_epi_relax}) has integer solutions, then \emph{strong duality} holds for the original MILP~(\ref{eq:mini_max_primal_problem_epi}).
Thus, it is natural to investigate the conditions, under which the LP~(\ref{eq:mini_max_primal_problem_epi_relax}) has integer solutions. Roughly speaking, there are not many results that establish conditions on LPs, beyond \emph{total unimodularity}~\cite[\S~9]{Truemper-98} of the associated problem matrices, under which they have integer solutions. Unfortunately, particularized to LP~(\ref{eq:mini_max_primal_problem_epi_relax}), the related matrices are not total unimodular, and therefore the theoretical implications of total unimodularity does not applies~\cite[\S~9]{Truemper-98}. Nevertheless, \emph{Proposition~\ref{prop:dualoptimal_Vs_relaxed_Primal}} allows us to imagine special cases of MILP~(\ref{eq:mini_max_primal_problem_epi}), where we have strong duality. We now present two such examples.

\begin{figure}[t]
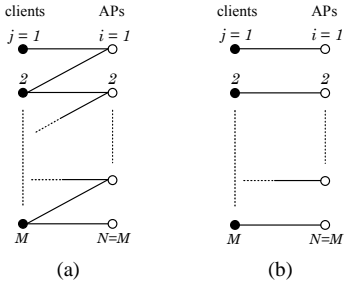

\centering
\subfigure[]
{\includegraphics[height=0.13\textheight]{zero_duality_ex1.epsi}
\label{fig:NW1}}
\goodgapp
\subfigure[]
{\includegraphics[height=0.13\textheight]{zero_duality_ex1_sol.epsi}
\label{fig:NW1_sol}}\vspace{-2mm}
\caption{Examples of network topology~1 (a)~The initial communication graph; (b)~The optimal association}
\label{fig:NW1_and_sol}
\vspace{-3.0mm}
\end{figure}

\begin{Ex}\label{ex:network1}
Suppose that the initial clients and APs communication is as shown in Figure~\ref{fig:NW1}. In particular, there are $M$ clients and $N=M$ APs. Each client $j\in\{2,3,\ldots,M\}$ can be associated to either AP~$j-1$ or $j$. However, the first client can be associated only to the AP~$1$. Moreover, suppose that $\beta_{jj}=\beta\in (0,1]$ for all $j\in\{1,\ldots,M\}$ and the remaining $\beta_{ij}$s can be arbitrary values from the range $(0,1]$.

We can easily show that the optimal client association (i.e., the optimal solution of MILP~(\ref{eq:mini_max_primal_problem_epi})) corresponds to Figure~\ref{fig:NW1_sol}. Moreover, we have $p^\star=\beta$. Let us now focus to the LP relaxation~(\ref{eq:mini_max_primal_problem_epi_relax}) applied to the network in Figure~\ref{fig:NW1}. In this case, every client~$j$, except the client~$1$, can be associated to both AP~$j-1$ \emph{and} AP~$j$. However, we can prove by contradiction: the solution of problem~(\ref{eq:mini_max_primal_problem_epi_relax}) corresponds to the optimal association depicted in Figure~\ref{fig:NW1_sol}. Thus, we have $p^\star=p^\star_{\mathrm{relax}}$. From \emph{Proposition~\ref{prop:dualoptimal_Vs_relaxed_Primal}}, we have strong duality for MILP~(\ref{eq:mini_max_primal_problem_epi}), (i.e., $p^\star=d^\star$) and the proposed algorithm achieves $d^\star$.
\end{Ex}

\begin{Ex}\label{ex:network2}
Consider the network shown in Figure~\ref{fig:NW2}. There are $N$ APs and two types of clients connected to APs. The Type~1 clients can communicate only with a \emph{single} AP. For example each client~$j\in\{1,\ldots,J_1\}$ can communicate only with AP~$1$, and therefore, they must be assigned to AP~$1$. On the other hand, the Type~2 clients can communicate with all the APs, e.g, clients $J_N+1,\ldots,M$. As a result, Type~2 clients can be associated to any AP. Now suppose that $\beta_{ij}$ values associated with Type~1 clients are such that $\sum_{j=1}^{J_1}\beta_{1j}=\sum_{j=J_1+1}^{J_2}\beta_{2j}\cdots =\sum_{j=J_{N-1}+1}^{J_N}\beta_{Nj}=B\in\R_+$ and $\beta_{ij}$ values associated with Type~2 clients are all equal to $\beta\in(0,1]$. Moreover, suppose that the number of Type~2 clients is a multiple of $N$, i.e., $M-J_N=mN$ for some $m\in\Z_+$.

The optimal client association (i.e., the optimal solution of MILP~(\ref{eq:mini_max_primal_problem_epi})) corresponds to Figure~\ref{fig:NW2_sol}, where the Type~2 clients are equally distributed among the APs. In particular, each AP is associated with $m$ Type~2 clients and we have $p^\star=B+m\beta$.
Let us now consider solution given by the LP relaxation~(\ref{eq:mini_max_primal_problem_epi_relax}). note that there is no choice for Type~1 clients, other than associating them to the only AP they can communicate. From the symmetry, we can easily see that associating each Type~2 client~$j$ among all the APs $i=1,\ldots,N$ with equal shares is a particular solution to the LP relaxation~(\ref{eq:mini_max_primal_problem_epi_relax}), i.e., for every Type~2 client $j$, $x_{ij}=(1/N)$, $i=1,\ldots,N$. Thus, at each AP~$i$, the utilization corresponds to Type~1 clients is $B$ and the utilization corresponds to Type~2 clients becomes $(M-J_N)(1/N)$, which is identical to $m$ (recall $M-J_N=mN$). Therefore, we have $p^\star_{\mathrm{relax}}=(B+m)=p^\star$ and strong duality holds for MILP~(\ref{eq:mini_max_primal_problem_epi}).
\end{Ex}

\begin{figure}[t]
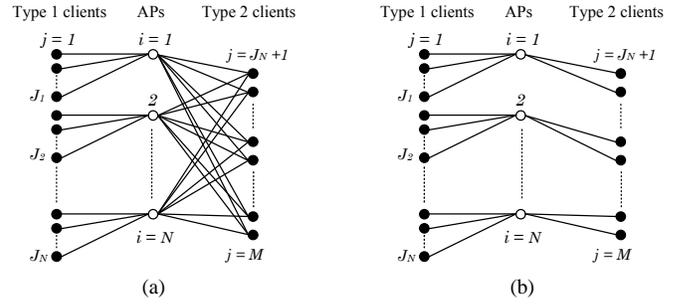

\centering
\subfigure[]
{\includegraphics[height=0.14\textheight]{zero_duality_ex2.epsi}
\label{fig:NW2}}
\goodgapp
\subfigure[]
{\includegraphics[height=0.14\textheight]{zero_duality_ex2_sol.epsi}
\label{fig:NW2_sol}}\vspace{-2mm}
\caption{Examples of network topology~2 (a)~The initial communication graph; (b)~The optimal association}
\label{fig:NW2_and_sol}
\vspace{-2.0mm}
\end{figure}

The examples above give insights into the algorithm's behavior in spacial cases. However, they can be used to build intuitive ideas of general networks. It is indeed important to analyze the proposed algorithm's behavior in general as well. In the sequel, we provide theoretical substantiation that allows us to predict the general algorithm properties in terms of the optimal duality gap and the relative duality~gap.

The following theorem formally establishes a bound on the duality gap and the asymptotic optimality of \emph{DAA}.

\begin{theorem}\label{thm:duality_gap_tends_to_zero}
The optimal duality gap of mixed integer linear program~(\ref{eq:mini_max_primal_problem_epi}) is bounded as follows:
\be\label{eq:bound_last_inside_theorem}
p^\star-d^\star \leq (N+1)(\varrho + \max_{j\in\mathcal{M}}\varrho_j) \ ,
\ee
where $\varrho=\max_{i\in\mathcal{N},{j}\in \mathcal{M}_i}\beta_{i{j}}$ and $\varrho_j=\min_{i\in\mathcal{N}_j}\beta_{ij}$. Moreover, the \emph{relative} duality gap diminishes to $0$ as $M\rightarrow\infty$.
\end{theorem}
\begin{proof}
See Appendix~\ref{app:Proof_Th_relative_zero_duality}.
\end{proof}
The theorem above suggests that the duality gap is always bounded by a \emph{constant} that does not depend on the number of clients in the system. Note that the bound grows like $N$, and therefore we can expect an increase in the duality gap for larger $N$. The theorem states that for larger $M$, the relative duality gap become almost zero. These are very important to get an insight of the behavior of the proposed algorithm in general networks, see \S~\ref{sec:numerical_results} for numerical examples.

\section{Numerical Examples}\label{sec:numerical_results}
In this section we present the numerical evaluation of the proposed algorithm in a multi-user multi-cell environment. We compare \emph{DAA} to: a)~random association policy, b)~RSSI-based policy, which is the association mechanism used in standards, and c)~optimal solution of the optimization problem~(\ref{eq:mini_max_primal_problem_epi}) using IBM CPLEX optimizer~\cite{CPLEX}.

\begin{figure}[t]
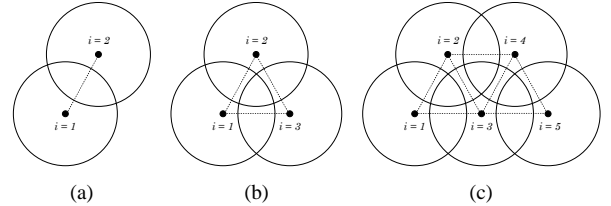

\centering
\subfigure[]
{\includegraphics[height=0.09\textheight]{top1.epsi}
\label{fig:topology1}}
\subfigure[]
{\includegraphics[height=0.09\textheight]{top2.epsi}
\label{fig:topology2}}
\subfigure[]
{\includegraphics[height=0.09\textheight]{top3.epsi}
\label{fig:topology3}}\vspace{-2mm}
\caption{Example simulation topologies (a)~2APs; (b)~3APs; (c)~5APs}
\label{fig:topology_all}
\vspace{-0.4cm}
\end{figure}
We define the SNR operating point at a distance $d$ [distance units] form any AP as
\begin{equation} \label{eq:SNR}\nonumber
\texttt{SNR}(d) = \left\{ \begin{array}{ll}
  \displaystyle{{P}_{0}}\lambda^2/(16\pi^2N_0W) & d\leq d_0\\
 \displaystyle{{P}_{0}}\lambda^2/(16\pi^2N_0W)\cdot\left({d}/{d_0}\right)^{-\eta} &  \textrm{otherwise}\ .
   \end{array} \right.
\end{equation}
Circular cells as depicted in~Figure~\ref{fig:system_example} are considered, where the radius of each cell $r$ is chosen such that $\texttt{SNR}(r)=10$~dB. The APs are located such that the distance between any consecutive APs is $D=1.1r$. For example, Figure~\ref{fig:topology1} shows the case for $N=2$ APs, Figure~\ref{fig:topology2} shows the case for $N=3$, and Figure~\ref{fig:topology3} shows the case for $N=5$. The clients are uniformly distributed among the circular cells and the potential AP-client association (i.e., $\mathcal{M}_i$ and $\mathcal{N}_j$) is found as pointed out in Figure~\ref{fig:system_example}.

We set $\lambda=5\:$mm, $N_0=-134\:$dBm/MHz, $W=1200\:$MHz, and $d_0=1\:$m, see~(\ref{eq:AWGN}) and~(\ref{eq:gain_coefficient}). Moreover, for all $j\in\mathcal{M}$, we set $I_j=0$ and for all $i\in\mathcal{N}, j\in\mathcal{M}_i$, we set $P_{ij}=P_0=0.1\:$mW and $G^{\mathrm{Tx}}_{ij}=G^{\mathrm{Rx}}_{ij}=1$. In order to check the average performance of the algorithms, we consider $\bar T=1000$ time slots, where the fading coefficients $\alpha_{ij}$ for all $i\in\mathcal{N}, \ j\in\mathcal{M}_i$ and the demanded data rates $Q_j$ for all $j\in\mathcal{M}$ are constant during each time slot $T\in\{1,\ldots,\bar{T}\}$ and independently change from one slot to another. In particular, the exponential random variables $\alpha_{ij}$ with unit mean are independent and identically distributed over the time slots. Moreover, we assume that $Q_j$ are uniformly distributed on $[0,400]$Mbits/s and independent and identically distributed over the time slots.

\begin{figure}[t]
\centering
\subfigure[]
{\includegraphics[height=0.22\textheight]{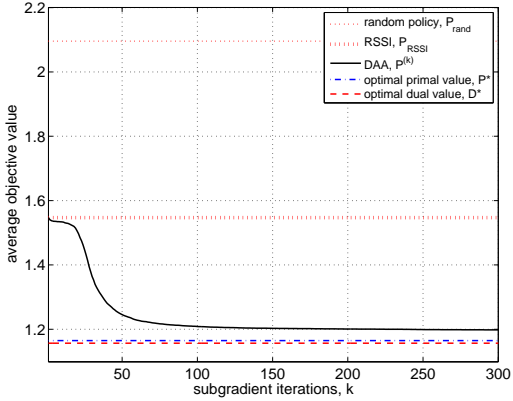}
\label{fig:Bound_illustration0}}
\goodgap
\subfigure[]
{\includegraphics[height=0.22\textheight]{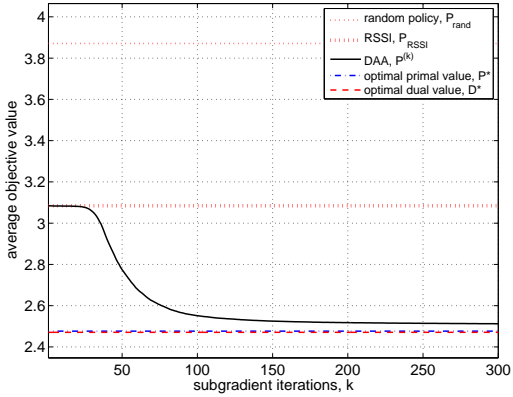}
\label{fig:convergence_illustration0}}\vspace{-2mm}
\caption{Influence of the number of clients on the convergence (a)~Average objective value $\texttt{P}^{(k)}$ vs. iterations $k$, $5$ APs, $100$ clients; (b)~Average objective value $\texttt{P}^{(k)}$ vs. iterations $k$, $5$ APs, $200$ clients}
\label{fig:Bound_illustration_all0}
\vspace{-0.4cm}
\end{figure}

\begin{figure}[t]
\centering
\subfigure[]
{\includegraphics[height=0.22\textheight]{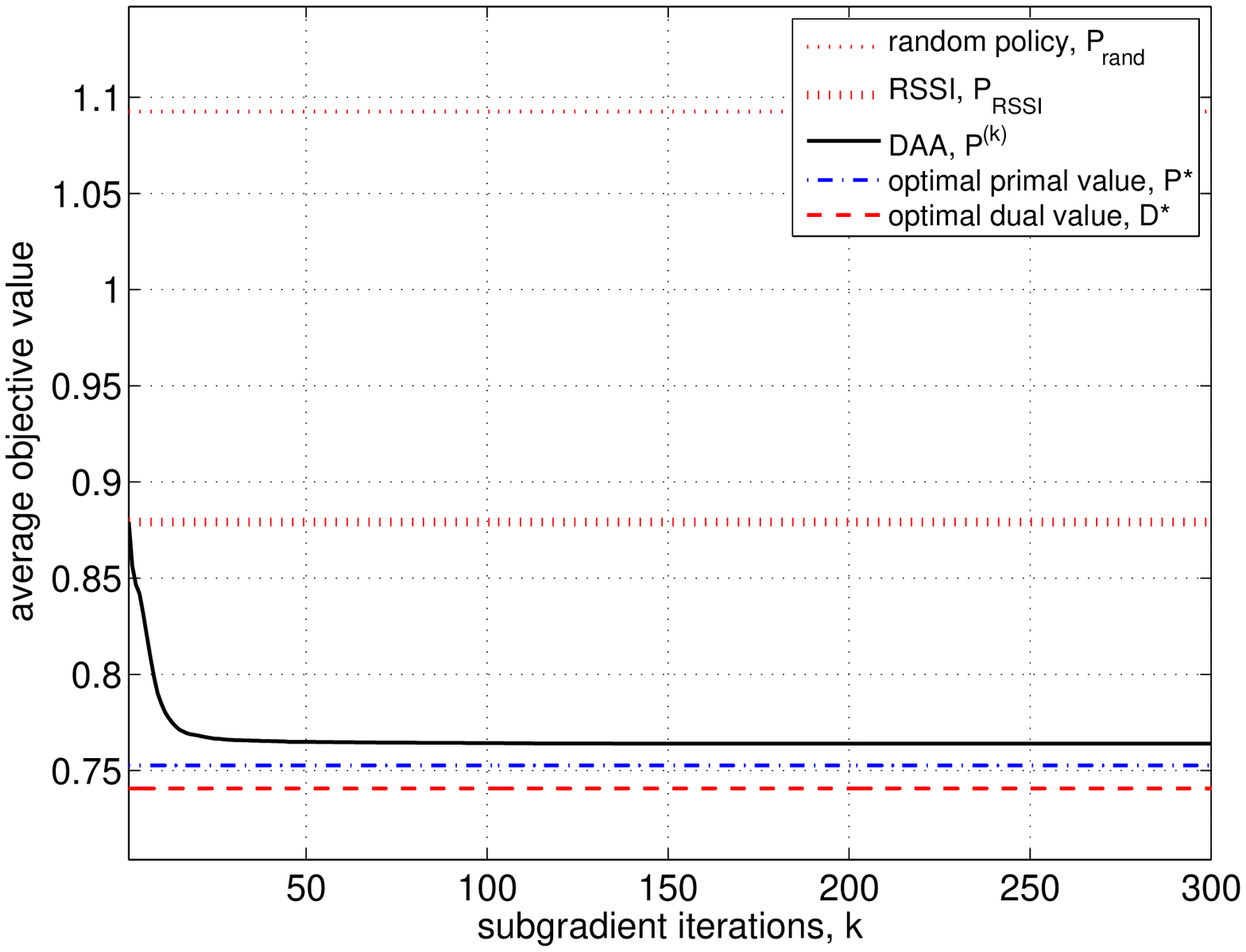}
\label{fig:Bound_illustration1}}
\goodgap
\subfigure[]
{\includegraphics[height=0.22\textheight]{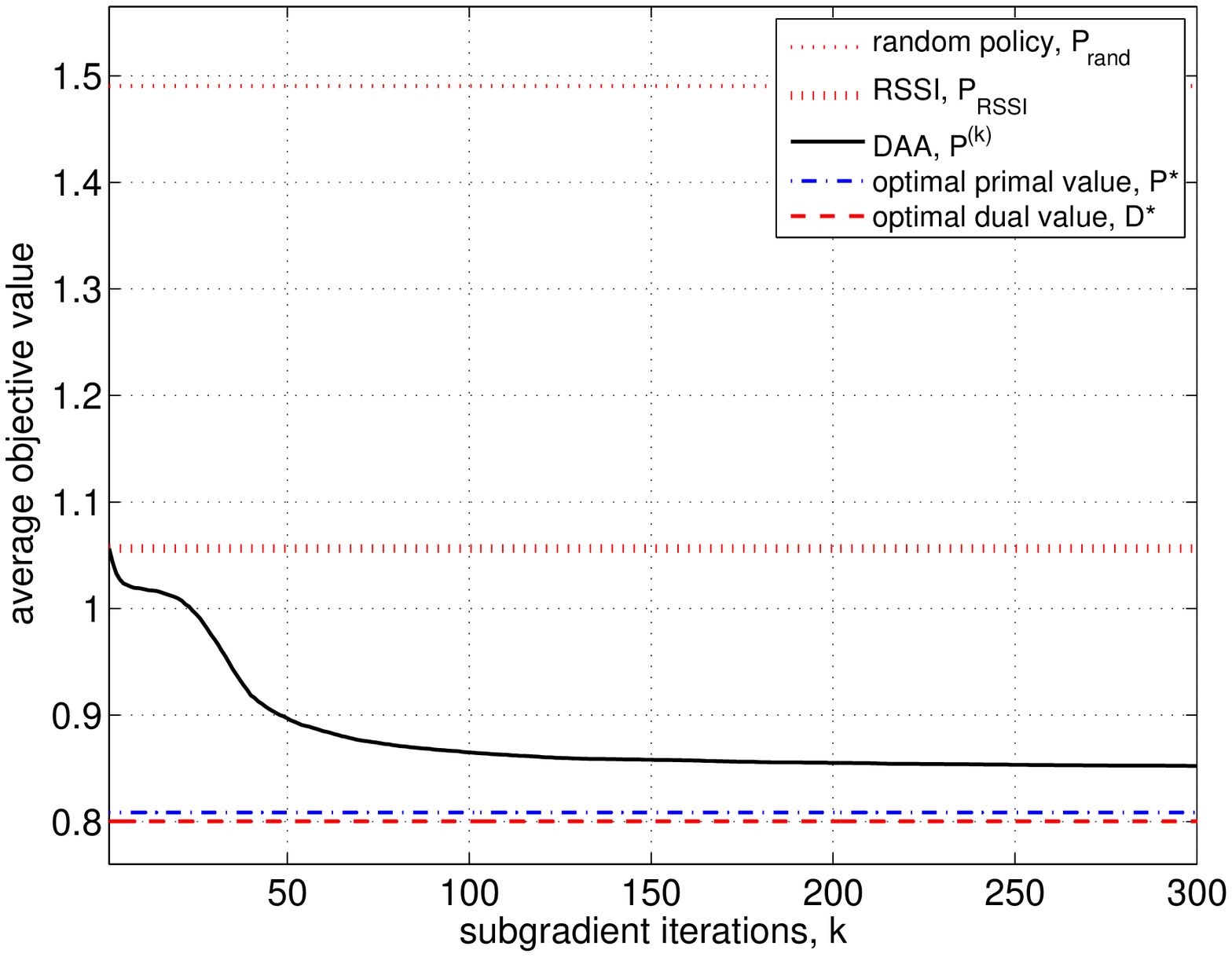}
\label{fig:convergence_illustration1}}\vspace{-2mm}
\caption{Influence of the number of APs on the convergence. (a)~Average objective value $\texttt{P}^{(k)}$ vs. iterations $k$, $3$ APs, $30$ clients; (b)~Average objective value $\texttt{P}^{(k)}$ vs. iterations $k$, $10$ APs, $100$ clients}
\label{fig:Bound_illustration_all1}
\vspace{-0.3cm}
\end{figure}

\begin{figure}[t]
\centering
\subfigure[]
{\includegraphics[height=0.22\textheight]{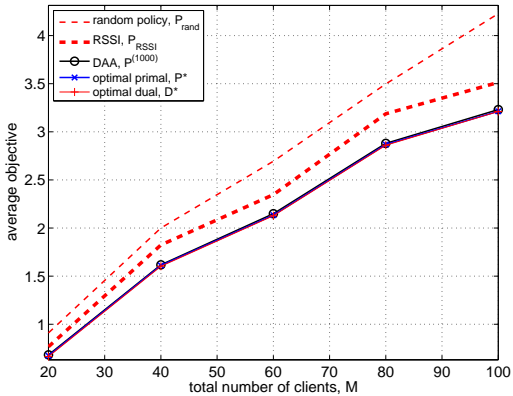}
\label{fig:Bound_illustration2}}
\goodgap
\subfigure[]
{\includegraphics[height=0.22\textheight]{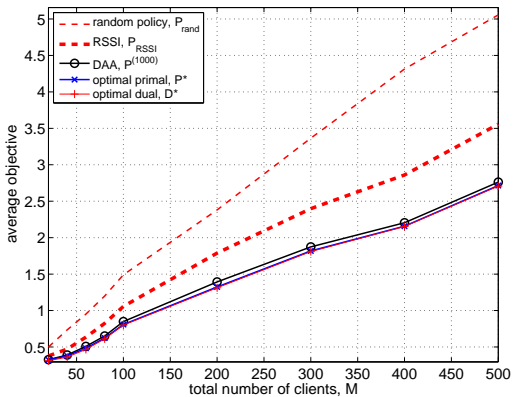}
\label{fig:convergence_illustration2}}\vspace{-2mm}
\caption{Average objective values $\texttt{P}_{\texttt{rand}}, \texttt{P}_{\texttt{RSSI}}, \texttt{P}^\star, \texttt{P}^{(1000)}$, and average dual optimal $\texttt{D}^\star$  vs. the number of users $M$. (a)~$2$ APs; (b)~$10$ APs}
\label{fig:Bound_illustration_all2}
\vspace{-0.4cm}
\end{figure}

\begin{figure}[t]
\centering
\subfigure[]
{\includegraphics[height=0.22\textheight]{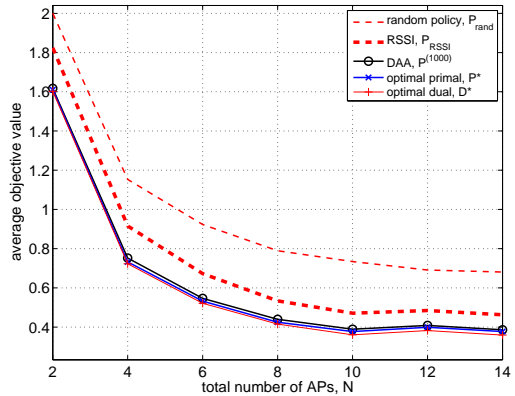}
\label{fig:Bound_illustration3}}
\goodgap
\subfigure[]
{\includegraphics[height=0.22\textheight]{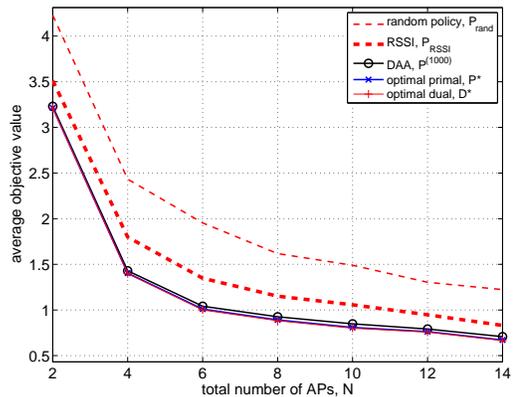}
\label{fig:convergence_illustration3}}\vspace{-2mm}
\caption{Average objective value $\texttt{P}_{\texttt{rand}}, \texttt{P}_{\texttt{RSSI}}, \texttt{P}^\star, \texttt{P}^{(1000)}$, and average dual optimal $\texttt{D}^\star$ vs. the number of APs $N$. (a)~$40$ users; (b)~$100$ users}
\label{fig:Bound_illustration_all3}
\vspace{-0.4cm}
\end{figure}

To see the average convergence behavior of the proposed algorithm, we consider the average \emph{primal} objective value of problem~(\ref{eq:mini_max_primal_problem_epi}) obtained by \emph{DAA}. In particular, the average primal objective value from \emph{DAA} after $k$ subgradient iterations is defined as $\texttt{P}^{(k)} = (1/\bar T)\sum_{T=1}^{\bar T}p^{(k)}_{\mathrm{best}}(T)$, where $p^{(k)}_{\mathrm{best}}(T)$ is the best \emph{primal feasible} objective value of problem~(\ref{eq:mini_max_primal_problem_epi}) after $k$ iterations at time slot $T$, see~(\ref{eq:best_primal_objective_afrer_k_itr}).\footnote{Due to the nonconvexity of the original problem~(\ref{eq:mini_max_primal_problem}), $p^{(k)}_{\mathrm{best}}(T)$ is not necessarily achieves the optimal value even when $k\rightarrow\infty$.} The average objective values from benchmark algorithms, random association policy, RSSI policy, and the optimal policy are defined in a similar manner, and are denoted by $\texttt{P}_{\texttt{rand}}$, $\texttt{P}_{\texttt{RSSI}}$, and $\texttt{P}^\star$, respectively.\footnote{Since the benchmark algorithms do not depend on subgradient iteration~$k$, like $\texttt{P}^{(k)}$, there is no superscript $(k)$ required for $\texttt{P}_{\texttt{rand}}$, $\texttt{P}_{\texttt{RSSI}}$, and $\texttt{P}^\star$.} Moreover, the average dual optimal value obtained by \emph{DAA} is $\texttt{D}^\star= (1/\bar T)\sum_{T=1}^{\bar T}d^\star(T)$, where $d^\star(T)$ is the optimal objective value of dual problem~(\ref{eq:dual_problem}) at time slot~$T$.

Figure~\ref{fig:Bound_illustration_all0} shows $\texttt{P}^{(k)}$ versus subgradient iteration~$k$, for the cases where $N=5$, $M=100$ (Figure~\ref{fig:Bound_illustration0}) and $N=5$, $M=200$ (Figure~\ref{fig:convergence_illustration0}). Results show that there is a noticeable effect of varying $M$ (the number of clients in the system) on convergence time. In particular, the convergence is faster for smaller~$M$. This observation is consistent with our analytical study presented in~\S~\ref{subsec:Convergence}. Proposed \emph{DAA} clearly outperforms the RSSI policy used in 802.11, 802.15.3c and 802.11ad, as well as the random policy. For example, \emph{DAA}, yields a performance improvement of about $20\%$ compared to RSSI policy in both considered cases. The gap between $\texttt{P}^\star$ and $\texttt{P}^{(k)}$, even after relatively larger the number of subgradient iterations (e.g., $k=300$) is indeed expected due to the nonconvexity of the original problem~(\ref{eq:mini_max_primal_problem}), see \S~\ref{subsec:RecoverPrimal}. Nevertheless, average dual optimal value $\texttt{D}^\star$ from \emph{DAA} is almost equal to the optimal $\texttt{P}^\star$.

Figure~\ref{fig:Bound_illustration_all1} shows $\texttt{P}^{(k)}$ versus subgradient iteration~$k$, for the cases where $N=3$, $M=30$ (Figure~\ref{fig:Bound_illustration1}) and $N=10$, $M=100$ (Figure~\ref{fig:convergence_illustration1}). Here, the clients density or the number of clients per AP is roughly the same (i.e., $10$). Results show that there is a clear effect of varying $N$ (while keeping client density fixed), on convergence. In particular, the convergence is faster for smaller~$N$. This observation is inline with our analytical study presented in~\S~\ref{subsec:Convergence}. The performance of other benchmark algorithms are very similar to those in Figure~\ref{fig:Bound_illustration_all0}.



Figure~\ref{fig:Bound_illustration_all2} shows the average objective from \emph{DAA} after $K=1000$ subgradient iterations, $\texttt{P}^{(K)}$ versus the number of clients $M$ for the cases $N=2$ (Figure~\ref{fig:Bound_illustration2}) and $N=10$ (Figure~\ref{fig:convergence_illustration2}). Plots for the benchmark algorithms are also depicted. Results show that the average objective values associated with each algorithm increase as $M$ increases. This is intuitively expected because APs become more loaded as the number of clients grows. Results further show that $\texttt{P}^{(K)}$ that came from \emph{DAA} is very close to the optimal $\texttt{P}^\star$ in both cases and the performance gap is not sensitive to changes in $M$. Note that the performance of the random and the RSSI policies are substantially low and their performance degradation becomes even noticeable for larger~$M$, see Figure~\ref{fig:convergence_illustration2}. As expected $\texttt{D}^\star$ provides a global lower bound on the performance, and is hardly distinguishable from $\texttt{P}^\star$.

Figure~\ref{fig:Bound_illustration_all3} shows the average objective value versus the number of APs $N$ for the cases where $M=40$ (Figure~\ref{fig:Bound_illustration3}) and $M=100$ (Figure~\ref{fig:convergence_illustration3}). The performance ranking of the algorithms is very similar to Figure~\ref{fig:Bound_illustration_all2}. Results show that the average objective values decrease as $M$ increases. This is intuitively explained by noting that, the larger the $N$ is, the smaller the client density is, and therefore the smaller the average objective values of each algorithm becomes. Results again shows that \emph{DAA} performs close to the optimal and outperforms the random and RSSI policies noticeably.


\begin{figure}[t]
\centering
\subfigure[]
{\includegraphics[height=0.22\textheight]{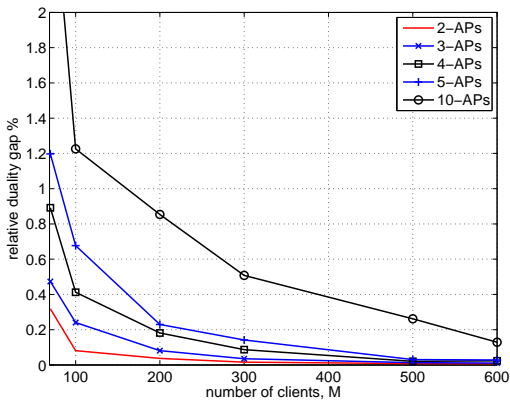}
\label{fig:Bound_illustration4}}
\goodgap
\subfigure[]
{\includegraphics[height=0.22\textheight]{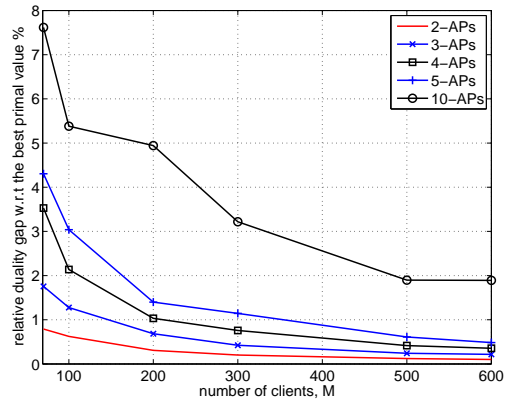}
\label{fig:convergence_illustration4}}\vspace{-2mm}
\caption{(a)~Average relative duality gap $\texttt{Ave-RDG}$ vs. the number of clients $M$; (b)~Average best achieved relative duality gap $\texttt{Ave-RDG-best-achieved}$ vs. the number of clients $M$}
\label{fig:Bound_illustration_all4}
\vspace{-0.4cm}
\end{figure}
\begin{figure}[t]
\centering
\subfigure[]
{\includegraphics[height=0.22\textheight]{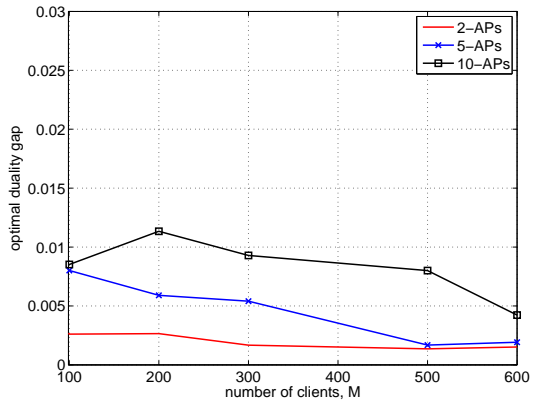}
\label{fig:Bound_illustration5}}
\goodgap
\subfigure[]
{\includegraphics[height=0.22\textheight]{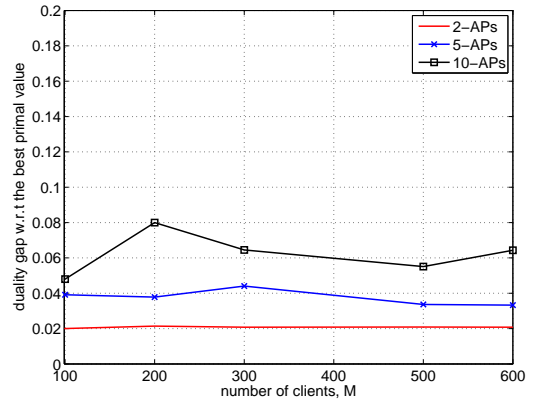}
\label{fig:convergence_illustration5}}\vspace{-2mm}
\caption{(a)~Average optimal duality gap $\texttt{Ave-DG}$ vs. the number of clients $M$; (b)~Best achieved average duality gap $\texttt{Ave-DG-best-achieved}$ vs. the number of clients $M$}
\label{fig:Bound_illustration_all5}
\vspace{-0.4cm}
\end{figure}

To see the effect of the increasing number of clients on the \emph{relative duality gap}~(see~\S~\ref{subsec:duality_gap}), we define the metric \emph{average relative duality gap}, $\texttt{Ave-RDG}$. In particular, $\texttt{Ave-RDG}= (1/\bar T)\sum_{T=1}^{\bar T}(p^\star(T)-d^\star(T))/p^\star(T)$, where $p^\star(T)$ is the optimal value of primal problem~(\ref{eq:mini_max_primal_problem_epi}) and $d^\star(T)$ is the optimal value of dual problem~(\ref{eq:dual_problem}), at time slot~$T$. Moreover, we denote by $\texttt{Ave-RDG-best-achieved}$, a related metric defined very similar to $\texttt{Ave-RDG}$, except that $p^\star(T)$ is replaced with $p^{(K)}_{\mathrm{best}}(T)$, i.e., the best \emph{primal feasible} objective value achieved from \emph{DAA} after $K$ iterations at time slot $T$.

Figure~\ref{fig:Bound_illustration4} measures the percentage $\texttt{Ave-RDG}$ versus $M$ for different $N$s. For all considered $\texttt{Ave-RDG}$, $N$ approaches to zero as $M$ increases. This is consistent with our analytical results established by \emph{Theorem~\ref{thm:duality_gap_tends_to_zero}}, see~\S~\ref{subsec:duality_gap}. Figure~\ref{fig:convergence_illustration4} shows that the plots of percentage $\texttt{Ave-RDG-best-achieved}$ versus $M$ are very
similar to those in Figure~\ref{fig:Bound_illustration4}. This behavior is not surprising, because the performance of \emph{DAA} is very close to the performance of the optimal method, see Figure~\ref{fig:Bound_illustration_all2} and Figure~\ref{fig:Bound_illustration_all3}.

Figure~\ref{fig:Bound_illustration_all5} depicts the dependence of the average duality gap on $M$. In particular, we define the \emph{optimal average duality gap} as $\texttt{P}^\star-\texttt{D}^\star$ and we plot \texttt{Ave-DG} versus $M$ are shown in Figure~\ref{fig:Bound_illustration5}. Moreover, we define the \emph{best achieved average duality gap} \texttt{Ave-DG-best-achieved} as $\texttt{P}^{(K)}-\texttt{D}^\star$. The corresponding plots are shown in Figures~\ref{fig:convergence_illustration5}. In both cases, results show that there is no apparent effect of the varying $M$ on the duality gap. Nevertheless, as we discussed in~\S~\ref{subsec:duality_gap} (see (\ref{eq:bound_last_inside_theorem})), the average duality gap grows when $N$~increases.

To examine the \emph{fairness} of the final client association among the APs, we consider the well known Jain's fairness index~\cite{Walrand00} as the fairness metric. We denote by $J^{(k)}(T)$ the fairness level resulted by the proposed \emph{DAA} at time slot $T$ after $k$ iterations. In particular, we define $J^{(k)}(T)= \big(\sum_{i\in\mathcal{N}}Y^{(k)}_i(T)\big)^2/(N\sum_{i\in\mathcal{N}}Y^{(k)}_i(T)^2)$, where $Y^{(k)}_i(T)=\sum_{j\in\mathcal{M}_i}\beta_{ij}x^{(k)}_{ij}(T)$ with $x^{(k)}_{ij}(T)$ being the best feasible solution resulted from \emph{DAA} at time slot $T$ and after $k$ iterations. The average fairness index $J^{(k)}$ resulted from \emph{DAA} after $k$ iterations is simply defined as $J^{(k)}=(1/\bar T)\sum_{T=1}^{\bar T}J^{(k)}(T)$. The average fairness indexes resulted from the benchmark algorithms, the random association policy, the RSSI policy, and the optimal policy are defined in a similar manner, and are denoted by $\texttt{J}_{\texttt{rand}}$, $\texttt{J}_{\texttt{RSSI}}$, and $\texttt{J}^\star$, respectively.\footnote{Benchmark algorithms do not depend on subgradient iteration~$k$. Therefore, like $\texttt{J}^{(k)}$, there is no superscript $(k)$ required for $\texttt{J}_{\texttt{rand}}$, $\texttt{J}_{\texttt{RSSI}}$, and~$\texttt{J}^\star$.}

Figure~\ref{fig:Bound_illustration_all6} depicts $J^{(k)}$ versus $k$ compared to the benchmark fairness indexes $\texttt{J}_{\texttt{rand}}$, $\texttt{J}_{\texttt{RSSI}}$, and $\texttt{J}^\star$ for the case where $N=5$ and $M=100$. Note that the fairness index ranges from $1/N$ (worst performance) to 1 (best performance). As expected, the optimal association provides the best performance. Results show that within a few hundreds of iterations, \emph{DAA} achieves a fairness level very close to the optimal. Results further show that \emph{DAA} significantly outperforms the random and the RSSI policies.

In order to provide a statistical description of the speed of
the proposed algorithm, we consider empirical the cumulative distribution function~(CDF) plots. Specifically, for each time slot $T\in\{1,\ldots,\bar{T}\}$, we store the total CPU time required for \emph{DAA} to find $p^{(K)}_{\mathrm{best}}(T)$. For comparison, we use the total CPU time required to find the \emph{optimal} value $p^\star(T)$. Figure~\ref{fig:Bound_illustration8} shows the empirical CDF plots of the number of iterations for $M=100,200,300$, with $N=10$. In the case of \emph{DAA}, the effect of changing the problem size by increasing $M$ on the CDF plots are almost indistinguishable. However, in the case of optimal method, there is a prominent increase in the time required to compute~$p^\star(T)$. Figure~\ref{fig:convergence_illustration8} depicts the average time required by \emph{DAA} and the optimal method versus $M$. Results show that the average time required by \emph{DAA} to find possibly a suboptimal solution is not sensitive to the variation of $M$, and is almost zero. However, the average time required by the optimal method to find the optimal solution grows approximately exponentially with $M$. This is certainly expected because problem~(\ref{eq:mini_max_primal_problem_epi}) is combinatorial, and therefore the worst-case complexity of the global method (CPLEX) grows exponentially with the problem size~\cite[\S~1.4.2]{Boyd-Vandenberghe-04}. Thus, there is naturally a tradeoff between the optimality and the efficiency of the algorithms. Nevertheless, Figures~\ref{fig:Bound_illustration_all2}, \ref{fig:Bound_illustration_all3} and \ref{fig:Bound_illustration_all8}, and the asymptotic results in Figure~\ref{fig:Bound_illustration_all4} indicate that \emph{DAA} yields a good tradeoff between the
optimality and the efficiency, which are favorable for practical implementation.

\begin{figure}[t]
\centering
{\includegraphics[height=0.22\textheight]{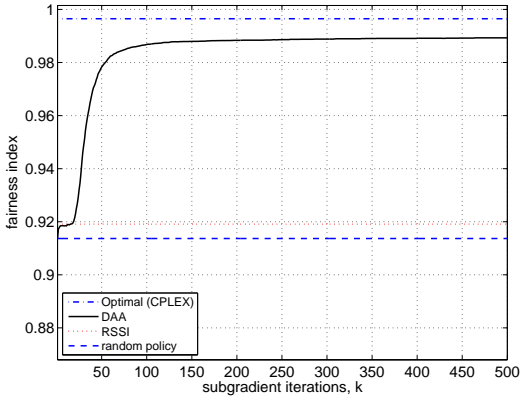}}\vspace{-2mm}
\caption{Average fairness index vs. the iterations $k$, $5$ APs, $100$ clients}
\label{fig:Bound_illustration_all6}
\vspace{-0.4cm}
\end{figure}

\section{Conclusions}\label{sec:conclusions}
In this paper we considered the problem of optimizing the allocation of the clients to the available APs in 60 GHz wireless access networks. The objective in our problem formulation was to \emph{minimize the maximum AP utilization} in the network. The optimization problem was combinatorial. Thus, we proposed a distributed association algorithm (\emph{DAA}), based on Langragian duality theory and subgradient methods. \emph{DAA} is fully compliant with the existing WiFi and 60 GHz protocols/standards and it can be easily implemented on top of the MAC mechanisms that they define. We studied the behavior of \emph{DAA} through theoretical analysis, where we proved its asymptotic optimality properties. Moreover, we presented a numerical analysis, where \emph{DAA} was compared to other association policies in realistic scenarios. We tested convergence, scalability, time efficiency, and fairness. Our results indicate that the proposed solution could be well applied in the forthcoming 60 GHz wireless access networks. 


\begin{figure}[t]
\centering
\subfigure[]
{\includegraphics[height=0.228\textheight]{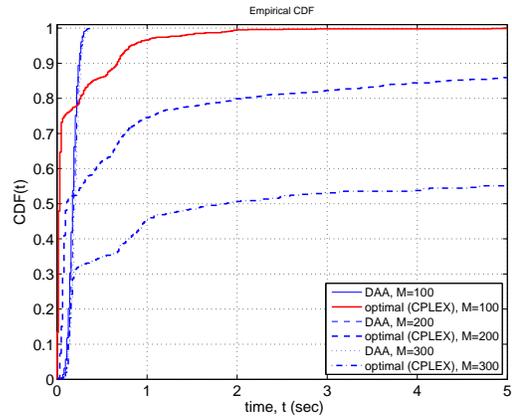}
\label{fig:Bound_illustration8}}
\goodgap
\subfigure[]
{\includegraphics[height=0.22\textheight]{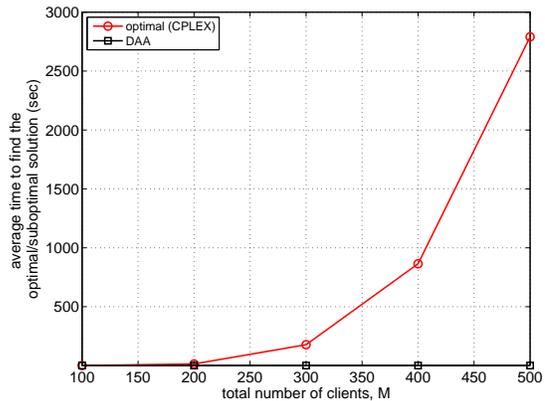}
\label{fig:convergence_illustration8}}\vspace{-2mm}
\caption{(a)~Empirical CDF plots of Total CPU time, $N=10$; (b)~Average time to find optimal/suboptimal solution, $N=10$}
\label{fig:Bound_illustration_all8}
\vspace{-0.4cm}
\end{figure}

\ifCLASSOPTIONcaptionsoff
  \newpage
\fi

\appendices
\input{appendix_singleColumn_V13}
\bibliographystyle{IEEEbib}
\bibliography{jour_short,conf_short,references,references_temp}
\end{document}

%% file: commands.tex
\newtheorem{theorem}{Theorem}
\newtheorem{prop}{Proposition}
\newtheorem{example}{Example}

\newcommand{\be}{\begin{equation}}
\newcommand{\ee}{\end{equation}}
\newcommand{\ba}{\begin{array}}
\newcommand{\ea}{\end{array}}
\newcommand{\bea}{\begin{eqnarray}}
\newcommand{\eea}{\end{eqnarray}}

\newcommand{\conv}{{\mbox{conv}}}

\newcommand{\tran}{^{\mbox{\scriptsize T}}}  

\newcommand{\asign}{{\mbox{$\colon\hspace{-2mm}=\hspace{1mm}$}}}

\newcommand{\vbar}{\raisebox{.17ex}{\rule{.04em}{1.35ex}}}
\newcommand{\vbarind}{\raisebox{.01ex}{\rule{.04em}{1.1ex}}}
\newcommand{\R}{\ifmmode {\rm I}\hspace{-.2em}{\rm R} \else ${\rm I}\hspace{-.2em}{\rm R}$ \fi}
\newcommand{\D}{\ifmmode {\rm I}\hspace{-.2em}{\rm D} \else ${\rm I}\hspace{-.2em}{\rm D}$ \fi}
\newcommand{\T}{\ifmmode {\rm I}\hspace{-.2em}{\rm T} \else ${\rm I}\hspace{-.2em}{\rm T}$ \fi}
\newcommand{\N}{\ifmmode {\rm I}\hspace{-.2em}{\rm N} \else \mbox{${\rm I}\hspace{-.2em}{\rm N}$} \fi}
\newcommand{\B}{\ifmmode {\rm I}\hspace{-.2em}{\rm B} \else \mbox{${\rm I}\hspace{-.2em}{\rm B}$} \fi}
\newcommand{\Hil}{\ifmmode {\rm I}\hspace{-.2em}{\rm H} \else \mbox{${\rm I}\hspace{-.2em}{\rm H}$} \fi}
\newcommand{\C}{\ifmmode \hspace{.2em}\vbar\hspace{-.31em}{\rm C} \else \mbox{$\hspace{.2em}\vbar\hspace{-.31em}{\rm C}$} \fi}
\newcommand{\Cind}{\ifmmode \hspace{.2em}\vbarind\hspace{-.25em}{\rm C} \else \mbox{$\hspace{.2em}\vbarind\hspace{-.25em}{\rm C}$} \fi}
\newcommand{\Q}{\ifmmode \hspace{.2em}\vbar\hspace{-.31em}{\rm Q} \else \mbox{$\hspace{.2em}\vbar\hspace{-.31em}{\rm Q}$} \fi}
\newcommand{\Z}{\ifmmode {\rm Z}\hspace{-.38em}{\rm Z} \else ${\rm Z}\hspace{-.28em}{\rm Z}$ \fi}

\renewcommand{\vec}[1]{\bf{#1}}     

\newenvironment{Ex}
{\begin{adjustwidth}{0.04\linewidth}{0cm}
\begingroup\small
\vspace{-0.10em}
\raisebox{-.20em}{\rule{\linewidth}{0.3pt}}
\begin{example}
}
{
\end{example}
\vspace{-2mm}
\rule{\linewidth}{0.3pt}
\endgroup
\end{adjustwidth}}


%% file: appendix_singleColumn_V13.tex


\section{Proof of Theorem~\ref{thm:duality_gap_tends_to_zero} }\label{app:Proof_Th_relative_zero_duality}
The proof is based on a proposition from~\cite{Bertsekas-96}, which we restate here for clarity and for simplifying the presentation.
\begin{prop}\label{prop:Bertsekas_estimate_of_duality_gap}
Consider the following possibly \emph{nonconvex} problem:
\begin{IEEEeqnarray}{lcl}\label{eq:Estimate_Duality_Problem}
\mbox{minimize} & \ \ & \textstyle\sum_{j\in\mathcal{J}}f_j({\vec y}_j)\IEEEyessubnumber\label{eq:Estimate_Duality_Problem1}\\
\mbox{subject to} & \ \  & {\vec y}_j\in\mathcal{Y}_j \ , \ j\in\mathcal{J} \IEEEyessubnumber\label{eq:Estimate_Duality_Problem2}\\
& \ \ & \textstyle\sum_{j\in\mathcal{J}}{\vec h}_j({\vec y}_j)\leq {\vec b} \IEEEyessubnumber\label{eq:Estimate_Duality_Problem3}  \ ,
\end{IEEEeqnarray}
where the variables are ${\vec y}_j\in\R^{y_j}$. The problem parameters $\mathcal{J}=\{1,\ldots,J\}$, ${\vec b}$ is a given vector in $\R^Q$, $\mathcal{Y}_j$ is a subset of $\R^{y_j}$, and $f_j:\conv{(\mathcal{Y}_j)}\rightarrow\R$ and ${\vec h}_j:\conv{(\mathcal{Y}_j)}\rightarrow\R^Q$ are functions defined on the convex hull of $\mathcal{Y}_j$. The following assumptions hold for the primal problem~(\ref{eq:Estimate_Duality_Problem}):

\noindent\textbf{Assumption 1:} There exist at least one feasible solution of problem~(\ref{eq:Estimate_Duality_Problem}).

\noindent\textbf{Assumption 2}: For each $j$, the subset of $\R^{y_j+Q+1}$
    \be
    \{({\vec y}_j,{\vec h}_j({\vec y}_j),f_j({\vec y}_j)) \ | \ {\vec y}_j\in\mathcal{Y}_j\}
    \ee
    is compact.

\noindent\textbf{Assumption 3:} For each $j$, given any vector $\tilde{\vec y}$ in $\conv{(\mathcal{Y}_j)}$, there exists ${\vec y}\in\mathcal{Y}_j$ such that ${\vec h}_j({\vec y})\leq \tilde{\vec h}_j(\tilde{\vec y})$, where $\tilde{\vec h}_j:\conv{(\mathcal{Y}_j)}\rightarrow\R^Q$ is the \emph{convexified} version of ${\vec h}_j$ on $\conv{(\mathcal{Y}_j)}$ and the notation ``$\leq$" here means the componentwise inequality. In particular, for all $\tilde{\vec y}\in\conv{(\mathcal{Y}_j)}$ 

\setcounter{equation}{33}
\begin{table*}[b!]
\normalsize
\vspace{-3mm}
\begin{tabular}{p{17.7cm}} \hline  \\ \end{tabular}
\vspace{-9mm}
\begin{subequations}\label{eq:d_star_equals_D_star}
\begin{align}\label{eq:d_star_equals_D_star1}
\hh\hh\tilde{\vec h}_j(\tilde{\vec y}) & = \mathop{\textstyle{\sum}}_{k}\alpha^k(\bar{\beta}_{1j}x^k_{1j}-t^k_j,\ldots,\bar{\beta}_{n_jj}x^k_{nj}-t^k_j,\ldots,\bar{\beta}_{Nj}x^k_{Nj}-t^k_j)\\ \label{eq:d_star_equals_D_star2}
& =\mathop{\textstyle{\sum}}_{k}(\bar{\beta}_{1j}\alpha^k x^k_{1j}-\alpha^k t^k_j,\ldots,\bar{\beta}_{n_jj}\alpha^k x^k_{nj}-\alpha^k t^k_j,\ldots,
\bar{\beta}_{Nj}\alpha^k x^k_{Nj}-\alpha^k t^k_j)\\ \label{eq:d_star_equals_D_star3}
& =\left(\bar{\beta}_{1j}\mathop{\textstyle{\sum}}_{k}(\alpha^k x^k_{1j})-\mathop{\textstyle{\sum}}_{k}(\alpha^k t^k_j),\ldots,\bar{\beta}_{n_jj}\mathop{\textstyle{\sum}}_{k}(\alpha^k x^k_{n_jj})-\mathop{\textstyle{\sum}}_{k}(\alpha^k t^k_j),\ldots,\bar{\beta}_{Nj}\mathop{\textstyle{\sum}}_{k}(\alpha^k x^k_{Nj})-\mathop{\textstyle{\sum}}_{k}(\alpha^k t^k_j)\right)\\ \label{eq:d_star_equals_D_star4}
& \geq\left(\hh\bar{\beta}_{1j}\mathop{\textstyle{\sum}}_{k}(\alpha^k x^k_{1j}){-}\hh\mathop{\textstyle{\sum}}_{k}(\alpha^k t^k_j),\ldots,\bar{\beta}_{n_jj}\mathop{\textstyle{\sum}}_{k}(\alpha^k  x^k_{n_jj}){-}\hh\mathop{\textstyle{\sum}}_{k}\big(\alpha^k (t^\mathrm{max}{+}\bar{\beta}_{n_jj}{x^k_{n_jj}})\big),\ldots,\bar{\beta}_{Nj}\hh\mathop{\textstyle{\sum}}_{k}(\alpha^k x^k_{Nj}){-}\mathop{\textstyle{\sum}}_{k}(\alpha^k t^k_j)\right)\\ \label{eq:d_star_equals_D_star5}
&\geq\left(\bar{\beta}_{1j}0-t^{\mathrm{max}}_c,\ldots,-t^{\mathrm{max}},\ldots,\bar{\beta}_{Nj} 0-t^{\mathrm{max}}_c\right)\\ \label{eq:d_star_equals_D_star6}
&=\left(-t^{\mathrm{max}}_c,\ldots,-t^{\mathrm{max}},\ldots,-t^{\mathrm{max}}_c\right)\\ \label{eq:d_star_equals_D_star7}
& = {\vec h}_j({\vec y}) \ ,
\end{align}
\end{subequations}
\end{table*}
\setcounter{equation}{27}

\begin{multline}\label{eq:h_tilde}
\tilde{\vec h}_j(\tilde{\vec y})= \textstyle\inf\bigg\{ \sum_{k=1}^{y_j+1}\alpha^k{\vec h}_j({\vec y}^k)  \left| \ \tilde{\vec y}=\sum_{k=1}^{y_j+1}\alpha^k{\vec y}^k,\right. \\
 \textstyle{\vec y}^k\in\mathcal{Y}_j, \sum_{k=1}^{y_j+1}\alpha^k=1, \alpha^k\geq 0   \bigg\} \ .
\end{multline}
Moreover, consider the dual problem of~(\ref{eq:Estimate_Duality_Problem}), i.e.,
\begin{equation} \label{eq:Estimate_Duality_Problem_dual}\nonumber
\begin{array}{ll}
\hspace{-2mm}\mbox{maxmize} & \hspace{-0mm}d(\boldsymbol\nu){=}\displaystyle\inf_{\substack{{\vec y}_j\in\mathcal{Y}_j\\j\in\mathcal{J}}}\textstyle\left\{\displaystyle\mathop{\textstyle{\sum}}_{j\in\mathcal{J}}  [f_j({\vec y}_j) + \boldsymbol\nu\tran{\vec h}_j({\vec y}_j)] -\boldsymbol\nu\tran{\vec b} \right\} \\
\hspace{-2mm}\mbox{subject to} & \hspace{-0mm}\boldsymbol\nu\geq{\vec 0} \ ,
\end{array}
\end{equation}
with variables $\boldsymbol\nu=(\nu_1,\ldots,\nu_Q)\in\R^Q$. Then we have
\be\label{eq:duality_gap_bertsekas}
P^\star - D^\star \leq (Q+1)\max_{j\in\mathcal{J}}{\rho_j} \ ,
\ee
where $P^\star$ denotes the optimal value of problem~(\ref{eq:Estimate_Duality_Problem}), $D^\star$ denotes the optimal value of the dual problem~(\ref{eq:Estimate_Duality_Problem_dual}), and $\rho_j$ is a nonnegative scalar such that
\be\label{eq:duality_gap_bertsekas_rho}
\rho_j \leq \sup_{{\vec y}_i\in\mathcal{Y}_j}f_j({\vec y}_j)- \inf_{{\vec y}_i\in\mathcal{Y}_j}f_j({\vec y}_j) \ .
\ee
\end{prop}
\begin{proof}
We do not reproduce the proof here, but refer the interested reader to~\cite[\S~5.6.1, pp.~371-376]{Bertsekas-96} for a rigorous proof and to~\cite[\S~5.1.6]{Bertsekas-99} for an intuitive explanation. 
\end{proof}

Now we rely on \emph{Proposition~\ref{prop:Bertsekas_estimate_of_duality_gap}} above to prove \emph{Theorem~\ref{thm:duality_gap_tends_to_zero}}. The key steps of the proof are: (a) We equivalently reformulate MILP~(\ref{eq:mini_max_primal_problem_epi}) in the form~(\ref{eq:Estimate_Duality_Problem}) and (b) We show that the \textbf{Assumptions~1-3} of \emph{Proposition~\ref{prop:Bertsekas_estimate_of_duality_gap}} hold for this equivalent problem.

Let us start by considering the following problem which is closely related to MILP~(\ref{eq:mini_max_primal_problem_epi}):
\begin{equation} \label{eq:mini_max_primal_problem_epi-s_equi}
\begin{array}{ll}
\mbox{minimize} & \sum_{j\in\mathcal{M}}t_j \\
\mbox{subject to} & \sum_{j \in \mathcal{M}} \bar{\beta}_{ij}\ x_{ij} \leq \sum_{j\in\mathcal{M}}t_j, \ i\in\mathcal{N} \\
& \sum_{i \in \mathcal{N}}\gamma_{ij} x_{ij}=1, \   j\in\mathcal{M} \\
&  x_{ij}\in\{0,1\}, \ j\in\mathcal{M}, i\in\mathcal{N} \\
&  0\leq t_{j}\leq t^{\mathrm{max}}+\bar{\beta}_{n_jj}x_{n_jj}, \ j\in\mathcal{M} \ ,
\end{array}
\end{equation}
where the variables are ${\vec t}=(t_1,\ldots,t_M)$ and ${\vec x}=(x_{ij})_{i\in\mathcal{N}, \ j\in\mathcal{M}}$. The problem $\bar{\beta}_{ij}$ and $\gamma_{ij}$ are defined as

\begin{equation} \label{eq:beta_bar_and_gamma_defn}\hspace{-2mm}\nonumber
\begin{array}{ll}
\bar{\beta}_{ij}= \left\{\hspace{-2mm} \begin{array}{ll}
  \beta_{ij} & \textrm{$i{\in}\mathcal{N}, \ j{\in}\mathcal{M}_i$}\\
  0 &  \textrm{otherwise}
   \end{array} \right.
\end{array}\hspace{-1mm},\hspace{0mm}
\begin{array}{ll}
\gamma_{ij}= \left\{\hspace{-2mm} \begin{array}{ll}
  1 & \textrm{$j{\in}\mathcal{M},\ i{\in}\mathcal{N}_j$}\\
  0 &  \textrm{otherwise}
   \end{array} \right.
\end{array}  ,
\end{equation}
$n_j=\arg\min_{i\in\mathcal{N}_j}\bar{\beta}_{ij}$, and $t^{\mathrm{max}}<\infty$ is an upper bound on $t^\star_{j}$, the optimal solution component of problem~(\ref{eq:mini_max_primal_problem_epi-s_equi}) that corresponds to $t_j$. For example, we use $t^{\mathrm{max}}=\max_{i\in\mathcal{N},\bar{j}\in \mathcal{M}}\bar{\beta}_{i\bar{j}}$, throughout this paper. We can easily show that problem~(\ref{eq:mini_max_primal_problem_epi-s_equi}) is equivalent to original MILP~(\ref{eq:mini_max_primal_problem_epi}) and the optimal value $P^\star$ is equal to the optimal value $p^\star$ of MILP~(\ref{eq:mini_max_primal_problem_epi}), i.e.,
\be\label{eq:P_star_equal_p_star}
P^\star =p^\star \ .
\ee
We refer to problem~(\ref{eq:mini_max_primal_problem_epi-s_equi}) as the \emph{modified MILP}, which is in the form~(\ref{eq:Estimate_Duality_Problem}), where
\begin{enumerate}
\item[1.] $\mathcal{J} = \mathcal{M}$ and $J= M$,
\item[2.] ${\vec y}_j=({\vec z}_j,t_j)\in\R^{y_i}$, with ${\vec z}_j=(x_{ij})_{i\in\mathcal{N}}$ and $y_j=N+1$,
\item[3.] $f_j({\vec y}_i)=t_j$,
\item[4.] $\mathcal{Y}_j=\big\{ ((x_{ij})_{i\in\mathcal{N}},t_j) \ \left| \ \sum_{i \in \mathcal{N}} \gamma_{ij}x_{ij}=1, \right. \  x_{ij}\in\{0,1\}, \ i\in\mathcal{N}, \ t_{j}\in[0,t^{\mathrm{max}}+\bar{\beta}_{n_jj}x_{n_jj}] \big\}$,
\item[5.] ${\vec h}_j({\vec y}_j)=((\bar{\beta}_{1j}x_{1j}-t_j),\ldots,(\bar{\beta}_{Nj}x_{Nj}-t_j))\in\R^Q$, with $Q=N$, and
\item[6.] ${\vec b}= {\vec 0}$.
\end{enumerate}

Let us now show that the \textbf{Assumptions~1-3} of \emph{Proposition~\ref{prop:Bertsekas_estimate_of_duality_gap}} hold for the modified MILP~(\ref{eq:mini_max_primal_problem_epi-s_equi}). It is straightforward to see that \textbf{Assumptions~1-2} hold. Checking whether \textbf{Assumption~3} holds is less trivial as we show next.

Let $\tilde{{\vec y}}$ be any given vector in $\conv{(\mathcal{Y}_j)}$. From the definition of $\tilde{\vec h}_j(\tilde{\vec y})$ (see~(\ref{eq:h_tilde})), we can express it as
\be
\tilde{\vec h}_j(\tilde{\vec y})= \sum_{k=1}^{y_j+1}\alpha^k{\vec h}_j({\vec y}^k)
\ee
for some ${\vec y}^k\in\mathcal{Y}_j$ and $\alpha^k$ such that $\tilde{\vec y}=\sum_{k=1}^{y_j+1}\alpha^k{\vec y}^k$, $\sum_{k=1}^{y_j+1}\alpha^k=1, \ \alpha^k\geq 0$. Particularized to the modified MILP~(\ref{eq:mini_max_primal_problem_epi}), we can write the set of equations~(\ref{eq:d_star_equals_D_star1})-(\ref{eq:d_star_equals_D_star7}),
where $t^{\mathrm{max}}_c = t^{\mathrm{max}}+\bar{\beta}_{nj}$ and ${\vec y}= (0,0,\ldots,1,0,\ldots,0,t^{\mathrm{max}}_c)$, which is feasible, i.e.,  ${\vec y}\in\mathcal{Y}_j$ and \textbf{Assumption~3} holds. Note that the first three equalities~(\ref{eq:d_star_equals_D_star1})-(\ref{eq:d_star_equals_D_star3}) follows from straightforward manipulations,  inequality~(\ref{eq:d_star_equals_D_star4}) follows from that $t^k_j\leq t^{\mathrm{max}}+\bar{\beta}_{n_jj}{x^k_{n_jj}}$ for all $k$, inequality~(\ref{eq:d_star_equals_D_star5}) follows from that $\mathop{\textstyle{\sum}}_{k}(\alpha^k x^k_{1j})\geq 0$ for all $k$, and last two equalities~(\ref{eq:d_star_equals_D_star6})-(\ref{eq:d_star_equals_D_star7}) follow from straightforward manipulations.

Finally, let us show that $D^\star=d^\star$, where $D^\star$ is the dual optimal value of the associated dual problem~[compare with~(\ref{eq:Estimate_Duality_Problem_dual})] of the modified MILP~(\ref{eq:mini_max_primal_problem_epi-s_equi}). We denote by $P^\star_{\mathrm{relax}}$ the optimal value of the LP relaxation of~(\ref{eq:mini_max_primal_problem_epi-s_equi}). Thus, we have \addtocounter{equation}{1}
\begin{equation}\label{eq:d_star_equals_D_star}
\begin{split}
D^\star & = P^\star_{\mathrm{relax}}= p^\star_{\mathrm{relax}}=d^\star \ ,
\end{split}
\end{equation}
where the first equality follows from a similar approach as described in the proof of \emph{Proposition~\ref{prop:dualoptimal_Vs_relaxed_Primal}}, the second equality follows from that the LP relaxations of problem~(\ref{eq:mini_max_primal_problem_epi-s_equi}) and of problem~(\ref{eq:mini_max_primal_problem_epi}) (see problem~(\ref{eq:mini_max_primal_problem_epi_relax})) have the same optimal value, and the last equality follows from \emph{Proposition~\ref{prop:dualoptimal_Vs_relaxed_Primal}}.

By using~(\ref{eq:P_star_equal_p_star}), (\ref{eq:d_star_equals_D_star}), and that \textbf{Assumptions~1-3} hold for the MILP~(\ref{eq:mini_max_primal_problem_epi}) together with \emph{Proposition~\ref{prop:Bertsekas_estimate_of_duality_gap}}, we have
\begin{equation}\label{eq:bound_last}
\begin{split}
p^\star-d^\star & \leq (N+1)\max_{j\in\mathcal{M}}\rho_j\\
& \leq (N+1)\max_{j\in\mathcal{M}}\big(t^{\mathrm{max}}+\bar{\beta}_{n_jj}\big)\\
& = (N+1)\left(\varrho + \max_{j\in\mathcal{M}}\varrho_j\right) \ ,
\end{split}
\end{equation}
where the first inequality follows from~(\ref{eq:duality_gap_bertsekas}), the second inequality follows from~(\ref{eq:duality_gap_bertsekas_rho}), $\sup t_j = t^{\mathrm{max}}+\bar{\beta}_{n_jj}$, and $\inf t_j = 0$, and the last equality follows from that $t^{\mathrm{max}}=\varrho$ and $\bar{\beta}_{n_jj}=\varrho_j$, which yields~(\ref{eq:bound_last_inside_theorem}).

%
We note from (\ref{eq:bound_last}) that the duality gap \emph{does not} depend on $M$, i.e., the total number of clients. Moreover, we note that $p^\star\rightarrow\infty$ as $m\rightarrow\infty$ (see the objective function of original problem~(\ref{eq:mini_max_primal_problem})). Thus, we conclude that, the \emph{relative} duality gap $(p^\star-d^\star)/p^\star$ diminishes to zero as $M\rightarrow\infty$.

%